%% file: main.tex
\documentclass[sigconf,nonacm,authorversion,dvipsnames,pdfa]{acmart}

\usepackage{colorprofiles}
\usepackage{graphicx}
\usepackage{caption}
\usepackage{amsmath}
\usepackage[vlined,commentsnumbered,ruled,linesnumbered]{algorithm2e}
\usepackage{algpseudocode}
\usepackage{verbatimbox}
\usepackage{float}
\usepackage{tcolorbox}
\usepackage{listings}
\usepackage[capitalize]{cleveref}
\usepackage{booktabs}
\usepackage{multirow}
\usepackage{lineno}
\usepackage{xspace}
\usepackage{colortbl}
\usepackage{mathtools}
\usepackage{enumitem}
\usepackage{comment}
\usepackage{standalone} 
\usepackage{etoolbox}
\providetoggle{long}
\usepackage{multicol}

\usepackage{makecell}
\usepackage{balance}

\usepackage{subfigure}

\newcommand\vldbdoi{10.14778/3705829.3705857}

\newcommand\vldbpagestyle{empty}

\input{macros}

\begin{document}
\setlength{\abovedisplayskip}{2.5pt}
\setlength{\belowdisplayskip}{2.5pt}

\title{Using Database Dependencies to Constrain \\ Approval-Based Committee Voting in the Presence of Context}

\settopmatter{authorsperrow=3}

\author{Roi Yona}
\orcid{0009-0005-5285-8687}
\affiliation{%
  \institution{Technion}
}
\email{roi.yona@campus.technion.ac.il}

\author{Benny Kimelfeld}
\affiliation{%
  \institution{Technion}
  \streetaddress{}
}
\email{bennyk@cs.technion.ac.il}

\begin{abstract}
In Approval-Based Committee (ABC) voting,  each voter lists the candidates they approve and then a voting rule aggregates the individual approvals into a committee that represents the collective choice of the voters. An extensively studied class of such rules is the class of ABC scoring rules, where each voter contributes to each possible committee a score based on the voter's approvals. We initiate a study of ABC voting in the presence of constraints about the general context surrounding the candidates. Specifically, we consider a framework in which there is a relational database with information about the candidates together with integrity constraints on the relational database extended with a virtual relation representing the committee. For an ABC scoring rule, the goal is to find a committee of maximum score such that all integrity constraints hold in the extended database. 

We focus on two well-known types of integrity constraints in relational databases: tuple-generating dependencies (TGDs) and denial constraints (DCs). The former can express, for example, desired representations of groups, while the latter can express conflicts among candidates.  ABC voting is known to be computationally hard without integrity constraints, except for the case of approval voting where it is tractable. We show that integrity constraints make the problem NP-hard for approval voting, but we also identify certain tractable cases when key constraints are used. We then present an implementation of the framework via a reduction to Mixed Integer Programming (MIP) that supports arbitrary ABC scoring rules, TGDs and DCs. We devise heuristics for optimizing the resulting MIP, and describe an empirical study that illustrates the effectiveness of the optimized MIP over databases in three different domains.
\end{abstract}

\maketitle

\pagestyle{\vldbpagestyle}

\ifdefempty{\vldbavailabilityurl}{}{
}

\input{section_introduction}

\input{section_preliminaries}
\input{section_formal_framework}
\input{section_complexity_analysis}

\input{section_implementation_via_ilp}
\input{section_experimental_evaluation}

\section{Concluding Remarks}
We presented a framework that extends ABC voting with a relational database  that provides contextual information about the candidates. The framework enables phrasing integrity constraints on the desired committee, with respect to the context, by deploying the classical concept of database dependencies. We focused on database dependencies in the form of TGDs and DCs. We presented a MIP implementation of the framework and devised optimization techniques.  Finally, we conducted an empirical study that shows the effectiveness of the optimized implementation. 

This work opens many directions for future research. For one, we would like to extend the set of constraints to ones that involve counting and other aggregations~\cite{DBLP:journals/tcs/RossSSS98,DBLP:journals/tods/FlescaFP10}, to reason about quantitative requirements about the sought committee. It would also be interesting to study the combination of our framework with proportionality in representation~\cite{DBLP:conf/aaai/AzizBCEFW15,DBLP:conf/aaai/AzizBCEFW15,DBLP:conf/aaai/0001EHLFS18,DBLP:conf/sigecom/Brill023}. In terms of the implementation strategy, we plan to study whether we can establish a comparable or more effective solution by deploying solver paradigms different from MIP, particularly Max-Sat solvers. 

Finally, we note that our framework treats the contextual database as mainly providing context about the candidates, while the voter profile is given as in the traditional ABC setting. Previous studies extended ABC by allowing the voters to phrase conditions on committees rather than mere approval sets~\cite{DBLP:conf/ijcai/BarrotL16, Boutilier_2004,ijcai2023p282, masařík2023generalisedtheoryproportionalitycollective,mavrov2023fairmultiwinnerelectionsallocation}. Accordingly, a valuable direction for future work is to complement past formalisms by extending our framework with the voter contexts, allowing voters to phrase conditions with respect to the contextual database (e.g., as DCs and TGDs over the extended schema).

\section*{Acknowledgments}
The authors are deeply grateful to Phokion Kolaitis for fruitful discussions and substantial suggestions for this work.

\bibliographystyle{ACM-Reference-Format}
\input{main.bbl}

\input{appendix}



\end{document}

%% file: macros.tex
\newcommand{\defeq}{\mathrel{\mathop:}=}

\def\scs{\mathcal{S}}

\def\comrel{\textsc{Com}}
\def\rel#1{\mathsf{#1}}
\def\scscom{\scs_+}
\def\dcom#1#2{#1[#2]}

\def\set#1{\mathord{\{#1}\}}

\def\score{\mathsf{score}}

\newcolumntype{R}{>{\columncolor[HTML]{D0D3D4}}l}
\newcommand{\att}[1]{{\normalfont\textrm{#1}}}
\newcommand{\val}[1]{{\normalfont\texttt{#1}}}

\newcommand{\eat}[1]{}

%% file: section_introduction.tex
\section{Introduction}

In Approval-Based Committee (ABC) voting, each voter states their approved set of candidates, and the goal is to aggregate all voter approvals into a choice of a winning committee (or several alternative winning committees) of some specified size~\cite{lackner2023multi,DBLP:conf/ijcai/LacknerS18}. The aggregation method is called an \emph{ABC rule}. In this work, we focus on a class of ABC rules that has been extensively studied, namely the ABC \emph{scoring rules}: 
every voter contributes a score to every set of candidates of the specified size, and a set of candidates is a winning committee if it accumulates the highest score;
moreover, the score of a voter is determined by two numbers: the number of committee members approved by the voter and the total number of candidates approved by the voter~\cite{DBLP:journals/jet/LacknerS21}. As a special case, in a \emph{Thiele rule} the score is determined only by the number of committee members approved~\cite{thiele}. The class of Thiele rules includes such extensively studied rules as (Proportional) Approval Voting ((P)AV) and Chamberlin-Courant (CC).

A winning committee should often satisfy criteria beyond having the highest voter support. For example, it may need to feature a set of skills so that it can properly function, must avoid conflicts among members,  must provide sufficient representation of different voter groups, must avoid over-representing of other groups, and so on. 
There is a significant body of work
on multi-winner elections (including ABC), where each candidate is associated with a set of attributes (e.g., political party and gender) and where  diversity constraints stating the legal quantities of occurrences for each attribute value can be phrased
\cite{DBLP:conf/aaai/BredereckFILS18, 
DBLP:journals/ai/LangS18,
DBLP:conf/ijcai/CelisHV18,DBLP:conf/ijcai/Relia22,masařík2023generalisedtheoryproportionalitycollective}.
\citet{DBLP:conf/aaai/AzizBCEFW15} initiated a line of research on ABC with Justified Representation (JR) that requires all voter sets that are \emph{large} and \emph{cohesive} to be \emph{well-represented} in the committee. The exact meaning of being large, cohesive, and well-represented differs between JR and stronger versions, such as Extended Justified Representation (EJR)~\cite{DBLP:conf/aaai/AzizBCEFW15}, Proportional Justified Representation (PJR)~\cite{DBLP:conf/aaai/0001EHLFS18}, and the more recent EJR+ and PJR+~\cite{DBLP:conf/sigecom/Brill023}. 
\citet{DBLP:conf/ijcai/YangW18} represented candidates as vertices of a graph and studied 
versions of multi-winner voting (different from ABC voting) that involve committees constrained to possess some subgraph property, such as being independent (e.g., no conflicts) or connected. 
\citet{masařík2023generalisedtheoryproportionalitycollective} take a combined approach that considers both diversity and fair representation, utilizing constraints with matroid structure. Similarly, \citet{ mavrov2023fairmultiwinnerelectionsallocation} study fairness properties of general constraints over the legal committees. In addition, committee voting can be viewed as a special case of \emph{participatory budgeting}, where the committee members can be seen as projects with different costs, and the goal is to decide on a set of projects (committee) that meets a given budget~\cite{cabannes2004participatory}; for this problem, various constraints have been studied, referring to dependencies between projects~\cite{DBLP:conf/ijcai/0001ST20,Rey2023AGF}.

The above rich literature on constraints in committee voting has been so far restricted to limited knowledge about the candidates, 
mainly in the form of labels (membership in groups). Consequently, the constraint formalisms have been restricted to low-level ones, such as direct relationships between committee members or cardinality requirements on label occurrences. 
Yet, the labels and relationships are naturally derived from contextual knowledge about the candidates, such as different types of relationships between candidates and between their associated entities. To use existing tools and voting algorithms, we need to translate the user's requirements to these formalisms. This leads to several challenges. First, it is not clear which types of user constraints are expressible in the low-level formalisms. Second, even if a translation exists, it might be inefficient. Third, by ignoring the patterns that lead to low-level constraints, we might lose opportunities for efficient algorithms that can rely on these patterns. 

In this work, we propose an approach to extend the treatment of constraints in committee voting by providing users with fragments of First-Order Logic to phrase committee constraints, referring to the whole contextual knowledge surrounding the candidates. To this aim, we initiate a study of ABC voting in the presence of contextual constraints by making a connection to database theory. 
In our framework, there is a relational database with information about the candidates; moreover,  there are integrity constraints on this relational database,  extended with a virtual relation $\comrel$ that represents the committee. Given an ABC scoring rule, the goal is to find a committee of maximum score such that all integrity constraints at hand are held in the extended database. 

We illustrate the framework by focusing on two fundamental and widely used types of integrity constraints in the context of relational databases: tuple-generating dependencies (TGDs)~\cite{DBLP:journals/jacm/BeeriV84} and denial constraints (DCs)~\cite{DBLP:conf/dagstuhl/BertossiC03}. The former can express, for example, representations of groups (or representation of a group under the condition that another group is represented), while the latter can express various conflicts among candidates.


It is well known that TGDs and DCs are expressible in first-order logic. As such, they cannot express diversity constraints that involve ranges of occurrences of attribute values
or constraints requiring that large and cohesive sets are well represented (actually, the latter seem to require the use of second-order logic). However, our framework supports a rich collection of constraints that go well beyond constraints on the attributes of a single relation, because TGDs and DCs can express complex relationships among candidates that involve multiple relations of a full database.


Our framework is akin to the framework of 
\emph{election databases}~  \cite{DBLP:conf/ijcai/KimelfeldKS18,DBLP:conf/pods/KimelfeldKT19}, which use the virtual $\rel{\textsc{Winner}}$ relation in single-winner elections with incomplete voter data.  There is, however, a major difference between the two frameworks: in these earlier papers, the database has no impact on the election, but it is used to query the possible outcomes; in contrast, here we use the database and the constraints to determine the winning committee (and defer the study of query answering in our framework to future work).

After describing the framework, we present a preliminary complexity analysis. It is known that (unconstrained) ABC voting is computationally hard for every Thiele rule, except for AV~\cite{DBLP:conf/atal/AzizGGMMW15,DBLP:journals/ai/SkowronFL16}. Not surprisingly,
integrity constraints make the problem NP-hard for AV. Yet, our complexity analysis illustrates that by being aware of the integrity constraints that hold in the external database, we can identify tractable special cases (e.g., by reduction to min-cost max-flow). We then develop an implementation of the framework through a reduction to Mixed Integer Programming (MIP) that applies to arbitrary ABC scoring rules, TGDs and DCs. 
We also present three heuristics for optimizing the MIP. Finally, we report on an empirical study on three different domains: political elections, hotels, and movies. We draw conclusions on the effectiveness of the optimizations, as well as the effect of various parameters on the performance of the MIP: the database constraints, the number of voters and candidates, the committee size,  and the scoring rule of choice.

%% file: section_preliminaries.tex
\section{Preliminaries}
We begin by introducing basic concepts and notation.



\paragraph{Approval-Based Committee voting.}
In Approval-Based Committee (ABC) voting, we have a set $C=\set{c_1,\dots,c_m}$ of candidates, a set $V=\set{v_1,\dots,v_n}$ of voters, an approval profile $A:V\rightarrow 2^C$ that maps every voter $v_i$ to a set $A(v_i)\subseteq C$ of candidates, 
and a desired committee size $k$. An \emph{ABC rule} determines which $k$-subsets of $C$ (i.e., sets of $k$ candidates) are the winning committees. As a special case, an \emph{ABC scoring rule}~\cite{DBLP:journals/jet/LacknerS21} is defined by a function $f(x,y)$ that determines the score contributed to a committee $B$ by a voter $v$ who approves a total of $y$ candidates where $x$ of them are in $B$. 
Formally, the winning committees by $f$ are the $k$-subsets $B$ of $C$ that maximize
$\score_f(B)\defeq\sum_{v\in V}f(|B\cap A(v)|,|A(v)|)$.
We make the usual assumption that 
$f(x,y)\geq f(x',y)$ if $x\geq x'$ 
(more approved members do not lower the score).

The class of ABC scoring rules generalizes the class of \emph{Thiele rules}~\cite{thiele} where $f(x,y)$ is a non-decreasing function $w(x)$ that depends only on the number of approved candidates in the committee.  Examples include \emph{Approval Voting} (AV) where $w(x) = x$ (i.e., the score contributed by a voter is the number of approved candidates), Proportional Approval Voting (PAV) where $w(x) = \sum_{i=1}^x{1/i}$ (i.e., the $i$th approved candidate increases the voter's score by $1/i$), and 
Chamberlin-Courant (CC) where $w(x) = \min(x,1)$ (i.e., the score is $1$ if the committee includes one or more approved candidates). An ABC scoring rule that is \emph{not} a Thiele rule is SAV (Satisfaction Approval Voting), where $f(x,y)=x/y$ (i.e., the number of approved committee members divided by the total number of candidates approved by the voter). 

It has been established that finding a winning committee is computationally hard for all Thiele rules, except for AV where the problem is straightforward. More precisely, after \citet{DBLP:conf/atal/AzizGGMMW15}
proved hardness for common Thiele rules, \citet[Theorem~5]{DBLP:journals/ai/SkowronFL16} proved a general result: for every Thiele rule it is NP-complete to decide whether there is a committee with a score above a given threshold, with the exception of the rules where $w(x)$ is linear in $x$ (i.e., rules equivalent to AV).


\paragraph{Relational databases.}
A \emph{relational schema} $\scs$ consists of a finite collection of \emph{relation symbols} $R$, each associated with an arity $k$. A \emph{database} $D$ over the schema $\scs$ associates a finite relation $R^D$ of $k$-tuples to each $k$-ary relation symbol $R$. A \emph{constraint} over a relational schema $\scs$ is a formula in some logical formalism (e.g., first-order logic) with relation symbols in $\scs$. We write $D\models\gamma$ to denote that the database $D$ satisfies the constraint $\gamma$. If   $\Gamma$ is a set of constraints, then  $D\models\Gamma$ denotes  that $D$ satisfies every constraint in $\Gamma$.
Next, we discuss two fundamental classes of constraints. 

An \emph{atomic formula} has the form $R(\tau_1,\dots,\tau_k)$ 
or of the form $\tau_1\theta\tau_2$, where $R$ is a $k$-ary relation symbol, each $\tau_i$ is either a constant or a variable,  and $\theta$ is a predefined comparison operator (e.g., $=$, $\leq$, $>$ and so on). We call 
$R(\tau_1,\dots,\tau_k)$ a \emph{relational atom} and $\tau_1\theta\tau_2$ a \emph{comparison atom}. 
A \emph{denial constraint} (DC) is an expression of the form
$$\forall\vec{x}[\neg(\varphi(\vec x)\land\psi(\vec x))]$$ 
where $\vec x$ is a sequence of variables, $\varphi(\vec x)$ is a conjunction of relational atoms, and $\psi(\vec x)$ is a conjunction of comparison atoms (where all variables belong to $\vec x$)~\cite{DBLP:conf/dagstuhl/BertossiC03}. Note that every \emph{key} constraint and, more broadly, every \emph{functional dependency} is a DC. In particular, if $R$ is a $k$-ary relation symbol, then by saying that  $i\in\set{1,\dots,k}$ is a \emph{key attribute} of $R$ we mean that the databases $D$ of interest must satisfy the DC asserting that no two distinct tuples in $R^D$ agree on the $i$-th attribute. 

A \emph{tuple-generating dependency} (TGD) has the form 
$$\forall\vec x[\varphi(\vec x) \rightarrow  \exists\vec y [\psi(\vec x,\vec y)]],$$ where $\vec x$ and $\vec y$ are disjoint sequences of variables and $\varphi(\vec x)$ and $\psi(\vec x,\vec y)$ are conjunctions of relational atoms~\cite{DBLP:journals/jacm/BeeriV84}. TGDs generalize common constraints like \emph{inclusion constraints}; for example, in the database of \Cref{fig:confrence_committee_db_example}, the DC $$\forall x,y[\rel{Author}(x,y) \rightarrow  \exists z [\rel{Pub}(y,z)]$$ requires every publication in the $\rel{Author}$ relation to occur in the $\rel{Pub}$ relation. Note that $\varphi(\vec x)$ can be an empty conjunction, hence a tautology (meaning that $\exists\vec y [\psi(\vec x,\vec y)]$ should hold unconditionally), and we denote such $\varphi(\vec x)$ by $\mathbf{true}$.

{
\def\vals#1{\val{\small #1}}
\begin{figure}[t]
\centering
\def\tabspace{\,\,\,}
\fbox{\parbox{0.3\textwidth}{
\begin{tabular}[b]{ll}
\multicolumn{2}{l}{Approvals:}  \\
$v_1$: & \vals{Ann}, \vals{Dave}    \\
$v_2$: & \vals{Ann}, \vals{Bob}, \vals{Dave}     \\
\end{tabular}
\begin{tabular}[b]{ll}
$v_3$: & \vals{Ann}, \vals{Eva}    \\
$v_4$: & \vals{Cale}              \\
$v_5$: & \vals{Bob}, \vals{Dave}   \\
\end{tabular}}}
\vskip1em
\renewcommand{\arraystretch}{1.0}
\begin{tabular}[t]{|l|}
    \multicolumn{1}{R}{$\rel{Topic}$}\\
    \hline
    \att{name} \\          
    \hline
    \vals{AI}   \\
    \vals{ML}   \\
    \vals{OS}   \\
    \vals{PL}   \\
    \hline
\end{tabular}\tabspace
\begin{tabular}[t]{|ll|}
    \multicolumn{2}{R}{$\rel{Supervise}$}\\
    \hline
    \att{advisor}                & \att{advised} \\
    \hline
    \vals{Ann}                    & \vals{Bob}    \\
    \vals{Bob}                    & \vals{Fred}   \\
    \vals{Cale}                   & \vals{Eva}   \\
    \vals{Dave}                   & \vals{Fred}   \\
    \hline
\end{tabular}\tabspace
\begin{tabular}[t]{|ll|}
    \multicolumn{2}{R}{$\rel{Author}$}\\
    \hline
    \att{author}                  & \att{pub} \\
    \hline
    \vals{Ann}                     & \vals{p1}   \\
    \vals{Ann}                     & \vals{p2}   \\
    \vals{Bob}                     & \vals{p1}   \\
    \vals{Bob}                     & \vals{p3}   \\
    \vals{Cale}                    & \vals{p4}   \\
    \vals{Dave}                    & \vals{p5}   \\
    \hline
\end{tabular}\tabspace
\begin{tabular}[t]{|ll|}
    \multicolumn{2}{R}{$\rel{Pub}$}\\
    \hline
    \att{pub}     & \att{topic} \\
    \hline
    \vals{p1}            & \vals{ML}  \\
    \vals{p2}            & \vals{PL}  \\
    \vals{p3}            & \vals{OS}  \\
    \vals{p4}            & \vals{AI}  \\
    \vals{p5}            & \vals{OS}  \\
    \hline
\end{tabular}\tabspace
\caption{ABC voting with external context.}
\label{fig:confrence_committee_db_example}
\vskip-1em
\end{figure}
}

%% file: section_formal_framework.tex
\section{ABC Voting in the Presence of  Constraints}\label{section:formal_framework}

Consider an approval profile $A$ over a set $C$ of candidates. Suppose that we have information about the candidates in a database $D$ over a schema $\scs$. We view $D$ as providing external context about the candidates. We would like to be able to express constraints on the desired committee, so that we restrict the collection of eligible committees to those that satisfy the given constraints. Towards this goal, we regard the constraints as ordinary database constraints on an \emph{extended} schema $\scscom$ defined to be the schema  $\scs$ augmented with a new relation symbol $\comrel/1$ that represents a hypothetical committee. (We assume that $\scs$ itself does not contain any relation symbol with the name $\comrel$.) Given a database $D$ and a set $B$ of candidates, we write $\dcom DB$ to denote the database over $\scscom$ that consists of $D$ and of the relation $B$ interpreting the relation symbol $\comrel$. Thus,
$R^{\dcom DB}\defeq R^D$ for every relation symbol $R$ of $\scs$, and 
$\comrel^{\dcom DB}\defeq B$.

An ABC setting with external context is a tuple of the form $(V,C,A,k,f,\scs,D,\Gamma)$ where,
$(V,C,A,k,f)$ is an ordinary ABC setting with the objective of finding a committee of size $k$ under the ABC scoring rule $f$;
$D$ is a database over the schema $\scs$; and
$\Gamma$ is a set of constraints over the schema $\scscom$.
A \emph{legal} committee is a $k$-subset $B$ of $C$ such that
$\dcom DB\models\Gamma$. A \emph{winning legal committee} is a legal committee $B$ such that $\score_f(B)$ is maximum among all legal committees $B'$. An  \emph{external constraint} is a constraint over $\scscom$.
This work focuses on DCs and TGDs.
\begin{example}\it\em
The approval profile of \Cref{fig:confrence_committee_db_example} (top) has $V=\set{v_1,\dots,v_5}$ and $C=\set{\val{Ann},\val{Bob},\val{Cale},\val{Dave},\val{Eva}}$. Assume that we seek a Program Committee (PC) for a conference. The relations in the figure give information about candidate publications and advisory relationships. The following DC states that the PC cannot include both a person and the advisor of that person.
\begin{align*}
\forall c_1,c_2 [
\neg\big(
    \rel{Supervise}(c_1,c_2) \land 
    \comrel(c_1)\land\comrel(c_2)
\big)
].
\end{align*}
For $k=3$ and AV, the committee $\set{\val{Ann}, \val{Bob}, \val{Dave}}$ is illegal, since \val{Ann} is the advisor of \val{Bob}. Note that every set of three candidates (among the five) is a legal committee, as long as it does not include both \val{Ann} and \val{Bob} or both \val{Dave} and \val{Eva}. The reader can verify that $B=\set{\val{Ann}, \val{Cale}, \val{Dave}}$ is a winning committee with $\score_{\mathrm{AV}}(B)=2+2+1+1+1=7$. 
\qed
\end{example}

\begin{example}\label{ex:tgd}\it\em
Continuing the running example, the TGD 
\begin{align*}
\forall t \big[ \rel{Topic}(t) \rightarrow \exists c,p 
    [\rel{Author}(c,p)\land \rel{Pub}(p,t)\land\comrel(c)]
\big]
\end{align*}
states that the committee has at least one member with a publication on each topic.

One can also include a specific TGD stating that at least one committee member should have both \val{ML} and \val{PL} publications (say, since the blend is central to the conference):
\begin{align*}
\textbf{true} \rightarrow\exists c,f,g \big[&\rel{Author}(c,f)\land\rel{Author}(c,g)\land\\
&\rel{Pub}(f,\val{ML})\land \rel{Pub}(g,\val{PL})\land\comrel(c)\big]
\end{align*}
For $k=3$ and AV as the voting rule, the reader can verify that the winning committee is $\set{\val{Ann}, \val{Bob}, \val{Dave}}$, since the three members cover all topics and, moreover, \val{Ann} published both \val{ML} and \val{PL} papers (namely \val{p1} and \val{p2}). Finally, the rule 
\begin{align*}
\forall c_1,c_2 \big[ 
&\rel{Supervise}(c_1,c_2) \land 
    \comrel(c_1)\land\comrel(c_2) \rightarrow \\ 
&    \exists p 
    [\rel{Author}(c_1,p)\land \rel{Pub}(p,\val{ML})]
\big]
\end{align*}
states that we allow for supervision between committee members only if the supervisor published in ML.\qed
\end{example}

%% file: section_complexity_analysis.tex
\section{Complexity Study}
We now discuss the complexity of ABC in the presence of constraints. We focus on AV, which is the only Thiele rule where there is hope for a tractability result, since finding a winning committee is NP-hard under every other Thiele rules~\cite{DBLP:conf/atal/AzizGGMMW15,DBLP:journals/ai/SkowronFL16}. 
Our analysis is carried out by considering specific patterns of schemas and constraints. We adopt the notion of  \emph{data complexity}~\cite{DBLP:conf/stoc/Vardi82}, where we assume that the schema and constraints are fixed and the input consists of the remaining components. In particular, every combination of a schema $\scs$ and set $\Gamma$ of constraints gives rise to a separate computational problem. We begin with a result about TGDs.
\def\thmtgdtxt{Let $t\geq 1$ be a fixed integer, and let $\scs$ be a schema consisting of the unary relation symbols $R_1,\dots,R_t$ and of the binary relation symbols $S_1,\dots,S_t$. Assume that $\Gamma$ consists of the  constraints $$\forall x\big[R_i(x)\rightarrow\exists c[S_i(c,x)\land \comrel(c)]\big]\,,$$
where $1\leq i\leq t$.
If $t\leq 2$ and if the first attribute of each $S_i$ is a key attribute, then a winning committee under AV can be found in polynomial time. 
In \emph{every} other case, it is NP-complete to determine whether a legal committee exists.
}
\begin{theorem}\label{thm:tgd}
\thmtgdtxt
\end{theorem}

The constraint $\forall x[R_i(x)\rightarrow\exists c[S_i(c,x)\land \comrel(c)]]$ asserts that for every $x$ in $R_i$, there is a committee member that relates to $x$ via $S_i$; for example, for every state $x$, at least one committee member $c$ lives in $x$.
The tractability part of the theorem can be proved by translating our framework to that of Bredereck et al.~\cite{DBLP:conf/aaai/BredereckFILS18}. Specifically, the combination of AV and the specific constraints for $\leq 2$ induce a \emph{diversity specification} for committee voting with labeled candidates, where the labeling is \emph{2-laminar} and the scoring function is \emph{separable}. In the Appendix, we give a direct proof by showing a greedy algorithm in the case of $t=1$, and a reduction to \emph{minimum-cost maximum flow}~\cite{DBLP:journals/mor/GoldbergT90} for $t=2$.

Hardness is proved by reductions from minimum set cover (in the absence of key constraints) and from exact matching by 3-sets (to account for key constraints). 
The message of \Cref{thm:tgd} is twofold. First, the problem of finding a legal committee of maximum score can be hard even for simple constraints, regardless of the voting rule. Second, being cognizant of the key constraints in the database can reduce the complexity of the problem and give rise to polynomial-time algorithms. The following states a result of a similar flavor for DCs.
\def\thmdctxt{
Let $\scs$ be a schema that contains a binary relation $R$, and assume that $\Gamma$ consists of the single DC
\begin{align*}
\forall c_1,c_2,x 
\big[&\neg(\comrel(c_1)\land\comrel(c_2)\land \\
& R(c_1,x)\land R(c_2,x)\land c_1\neq c_2)\big]
\end{align*}

If the first attribute of $R$ is a key attribute, then a winning committee under AV can be found in polynomial time. 
Otherwise, it is NP-complete to determine whether a legal committee exists.
}

\begin{theorem}\label{thm:dc}
\thmdctxt
\end{theorem}

The DC in  \Cref{thm:dc} asserts that no two distinct candidates can participate in the committee if they have a common neighbor in $R$ (e.g., they are affiliated with the same university). The tractability part is proved by reducing the problem to the unconstrained version of ABC with the AV rule. The hardness part is proved with a reduction from the maximum independent set problem.

%% file: section_implementation_via_ilp.tex
\section{Mixed Integer Programming Implementation}
As often done for intractable election problems in computational social choice~\cite{DBLP:journals/tdasci/ChakrabortyDKKR21,DBLP:conf/ijcai/DudyczMMS20,DBLP:conf/sigecom/Xia12,lackner2023multi,DBLP:journals/jossw/LacknerRK23}, we present a Mixed Integer Program (MIP)
to find a winning committee in the presence of TGDs and DCs.\footnote{Another possibility would have been to use a \emph{Max-SAT solver}, as suggested before on approval voting~\cite{DBLP:conf/ijcai/BarrotL16,DBLP:conf/ijcai/MarkakisP20}, though we are not aware of reported empirical results for this approach.} (Interestingly, this approach is also used in Computational Social Choice for theoretical upper bounds, typically Fixed Parameter Tractability~\cite{DBLP:conf/ijcai/BetzlerHN09,DBLP:conf/ecai/Yang14,DBLP:conf/pods/KimelfeldKT19}.)
 Recall that the problem instance has the form $(V,C,A,k,f,\scs,D,\Gamma)$ and that the goal that is to find a legal committee $B$ of maximum score.
Let $C=\set{c_1,\dots,c_m}$ and $V=\set{v_1,\dots,v_m}$. 

\subsection{Variables} The program uses the following variables.
\begin{itemize}[nosep]
\item $z_1,\ldots,z_m$ take values in $\set{0,1}$, where $z_j=1$ denotes that  $c_j$ is selected for the committee (i.e., $c_j\in B$).
\item $u_1,\dots,u_n$ take values in the domain $\mathbb{N}$ of the natural numbers, where $u_i$ gives the number $|B\cap A(v_i)|$ of candidates in the committee approved by voter $v_i$.
\item $s_1,\dots,s_n$ take non-negative real numbers, where $s_i$ 
gives the contribution $f(|B\cap A(v_i)|,|A(v_i)|)$ of voter $v_i$ to the committee. 
\end{itemize}

\subsection{Winning Committee without Constraints}
To find a winning committee without constraints, we formulate a MIP similar to the formulation proposed by \citet{DBLP:conf/ijcai/DudyczMMS20}.

\def\vskp{\vskip 1em}
\vskp\hrule
{
\begin{align}\label{eq:max} 
\mbox{Maximize\; } &\sum_{i=1}^n s_i \mbox{\; subject to:}\\
\forall v_i\in V: \quad & \!\!\!\!\sum_{c_j\in A(v_i)}\!\!\!\!\!z_j = u_i \notag
\quad\quad\vrule\quad\quad  \sum_{j=1}^m z_j=k \notag
\\
\forall v_i\in V,k'\in[k]: \quad & s_i\leq|k'-u_i|\cdot M+f(k',|A(v_i)|)\label{eq:absolute}
\end{align}
}
\hrule\vskp

Here, $M$ is a number greater than $f(k,|A(v)|)$ for all $v\in V$, and $[k]\defeq\set{0,\dots,k}$. We use a standard technique to eliminate the absolute value $|k'-u_i|$ from our MIP. 
Let $t=k'-y_i$. Note that $|t|<k+1$. 
Let $b\in\{0,1\}$, where $b=1$  denotes that $z\ge0$ (and $b=0$ denote $z\le0$), and let
$t^+,t^-\in\mathbb{Z}_{\ge 0}$. We add the constraint  $t=t^+-t^-$. We also add the constraint 
$0\leq t^+\leq b\cdot(k+1)$ to assure that $t^+=t$ if $b=1$; otherwise $t^+=0$. We use a symmetrical constraint $0\leq t^-\leq (1-b)\cdot(k+1)$ so that $t^-=t$ if $b=0$; otherwise $t^-=0$.
Finally, we replace $|k'-u_{i}|$ with $t^+ + t^-$.

\subsection{Incorporating TGDs}
Next, we add the TGDs of $\Gamma$ to the MIP. Consider a TGD $\tau\in\Gamma$.   
Since $\tau$ is a TGD over the extended schema $\scscom$, it is of the form
\begin{align*}
\forall \vec x \Big[ &
\big(\varphi'(\vec x)\land 
\comrel(x'_1)\land\ldots\land\comrel(x'_q)\big)\rightarrow\\
&\exists \vec y \big[
\psi'(\vec x,\vec y)\land 
\comrel(y'_1)\land\ldots\land\comrel(y'_\ell)
\big]\Big],
\end{align*}
where $x'_1,\dots,x'_q$ are distinct variables in $\vec x$, and $y'_1,\dots,y'_\ell$ are distinct variables in $\vec x$ and $\vec y$, and the expressions $\varphi'$ and $\psi'$ are conjunctions of relational atoms from $\scs$ (hence, they do not include the $\comrel$ relation).

Let $\mathcal{A}_\tau$ be the set of all 
assignments $\alpha$ of values from $D$ to the variables of $\vec x$ so that $\varphi'(\alpha(\vec x))$ holds in $D$, where $\alpha(\vec x)$ is the tuple obtained from $\vec x$ by replacing every variable $x$ with the value $\alpha(x)$. 
For $\alpha\in\mathcal{A}_\tau$, let $B[\alpha]\defeq\set{\alpha(x'_j)\mid j=1,\dots,q}$. For illustration, if $\tau$ is the first TGD of \Cref{ex:tgd}, then 
$\mathcal{A}_\tau$ will include the assignment $\alpha=\set{t\mapsto\val{ML}}$ since $\rel{Topic}(\val{ML})$ is true in the database of \Cref{fig:confrence_committee_db_example};
here, $B[\alpha]$ is empty since the TGD has no occurrences of $\comrel$ (but it would not be empty for the last TGD of  \Cref{ex:tgd}). 
Note that the premise of $\tau$ holds for $\alpha$ if and only if $B[\alpha]$ consists of only candidates, and each of them is in the committee. If $B[\alpha]$ includes a value that is not a candidate, we ignore $\alpha$. Hence, assume that 
$B[\alpha]\subseteq C$. We add to the MIP a variable $b_\alpha\in\set{0,1}$ that is intended to take value $1$ if and only if $B[\alpha]$ is a subset of the committee. 

\def\ext{\mathsf{ext}}

For $\alpha\in\mathcal{A}_\tau$,  let $\ext(\alpha)$ be the set of all extensions $\beta$ of $\alpha$ to an assignment (of values from $D$) to the variables of $\vec x$  and $\vec y$ such that $\psi'(\beta(\vec x),\beta(\vec y))$ holds in $D$.  Let $B[\beta]\defeq\set{\beta(y'_j)\mid j=1,\dots,\ell}$. For illustration, in our example  the assignment $$\beta=\set{t\mapsto\val{ML},p\mapsto\val{p1},c\mapsto\val{Ann}}$$ is in $\ext(\alpha)$, and $B[\beta]=\set{\val{Ann}}$. We add a variable $b_\beta \in\set{0,1}$ such that $b_\beta$ can be $1$ only if the conclusion of the TGD holds, that is,
$B[\beta]$ is a subset of the committee.

The following constraints ensure the correctness of the $b_\alpha$ and $b_\beta$, and also that every $\alpha$ has an extension $\beta$ that satisfies the conclusion of the TGD (hence the TGD holds).
\vskp\hrule
{
\begin{align*}
&\forall \alpha\in\mathcal{A}_\tau: \quad 
|B[\alpha]|\cdot b_\alpha 
\;\leq\;\!\!\!\!
\sum_{c_j\in B[\alpha]}\!\!\!\! z_j 
\;
\leq
\;
|B[\alpha]|+
b_\alpha-1 \\
&\forall \alpha\in\mathcal{A}_\tau, \beta\in\ext(\alpha): \quad 
|B[\beta]|\cdot b_\beta 
\;\leq\;
\sum_{c_j\in B[\beta]}\!\!\!\! z_j  \\
&\forall \alpha\in\mathcal{A}_\tau:\quad 
b_\alpha  
\;\leq\;
\sum_{\beta\in \ext(\alpha)}\!\!\!\! {b_\beta}
\end{align*}
}
\hrule\vskp

\subsection{Incorporating DCs}
Finally, we add the DCs of $\Gamma$ to the MIP. If $\delta$ is such a DC, then it is of the form 
\begin{align}\label{eq:delta}
\forall\vec{x}
[\neg(\varphi'(\vec x)\land\psi'(\vec x)\land \comrel(x'_1)\land\ldots\land\comrel(x'_q))],
\end{align}
where  $x'_1,\dots,x'_q$ are distinct variables in $\vec x$, and  the expressions $\varphi'$ and $\psi'$ are conjunctions of relational atoms from $\scs$.
Let $\mathcal{A}_\delta$ be the set of all assignments $\alpha$ such that 
$\varphi'(\vec x)\land\psi'(\vec x)$ holds in $D$. For $\alpha\in\mathcal{A}_\delta$, let $B[\alpha]=\set{\alpha(x'_j)\mid j=1,\dots,q}$. To enforce the satisfaction of the DC, we need to ensure that the committee does not contain any $B[\alpha]$. Hence, we add the following constraints to the MIP:
\vskp\hrule
{
\begin{align}\label{eq:dc-mip}
\forall \alpha\in\mathcal{A}_\delta: \quad 
\sum_{c_j\in B[\alpha]}\!\!\!\! z_j 
\;<\;
|B[\alpha]|
\end{align}
}
\hrule\vskp

\subsection{Optimizations}
We described a MIP for finding a winning committee under TGDs and DCs. As expected, the MIP may have a high execution cost due to the large number of variables and constraints. Next, we describe several optimizations of the MIP that we found highly beneficial in our experiments. These optimizations decrease the size of the program and have a twofold benefit: they reduce both the \emph{construction time}  and the \emph{solving time} of the MIP. We study the benefit of these optimizations in the empirical study of the next section.

\paragraph{Grouping similar voters.}
Voters may have identical approval profiles, especially if the number of candidates is relatively small. We group together candidates with the same profile and treat each group as a single \emph{weighted} voter $(v'_i,\mu'_i)$, where $\mu'_i$ is the number of voters $v$ grouped together for having the same approval $A(v)$, which becomes $A(v'_i)$. This means that $V$ consists now of the $v'_i$, and the program remains unchanged except for line~\eqref{eq:max} where we replace the objective $\sum_{i=1}^{n} s_i$ with $\sum_{i=1}^{n'} \mu_i\cdot s_i$.

\paragraph{Pruning infeasible scores.}
Recall the constraint~\eqref{eq:absolute} of the MIP. This constraint is applied to every $k'=0,\dots,k$ to force that $s_i$ takes the value $f(|B\cap A(v_i)|,|A(v_i)|)$. This rule gives a nontrivial inequality only when $k'=|B\cap A(v_i)|$. In particular, the rule is meaningless if $k'>|A(v_i)|$, and can be skipped. Hence, in the constraint~\eqref{eq:absolute} we replace $k$ with $\min(k, |A(v_i)|)$.
This optimization targets voters $v_i$ with $|A(v_i)|<k$, which we encountered frequently.

\paragraph{Contracting DC constraints via hypercliques.}
This optimization concerns the constraints for enforcing the DCs. For simplicity sake, suppose first that $q=2$ in the definition of $\delta$ in~\eqref{eq:dc-mip}. 
Define the \emph{conflict graph} of $\delta$ to be the undirected graph $G$ where $C$ is the set of nodes, and there is an edge between every pair of candidates that violates $\delta$. Then, the process of adding constraints can be viewed as iteratively eliminating an edge $e$ of $G$ by adding the constraint in \eqref{eq:dc-mip} for $e=B[\alpha]$, until no edges remain. The insight we use is that if $U$ is a clique of $G$, then we can eliminate all of the edges between vertices of $U$ with a single constraint:
$ \sum_{c_j\in U} < 2$.\footnote{This is assuming that there are no self-loops, which are easy to handle: simply ignore every candidate that has a self-loop.} With that, we contract multiple constraints into a single one. Hence, we can reduce the total number of constraints by iteratively eliminating maximal cliques $U$, instead of individual edges $e$, until no edges are left in $G$.
A common example where this optimization is useful is the case that $\delta$ forbids two candidates from the same party under some party definition (e.g., the parties of authors of publications on the same topic in  \Cref{fig:confrence_committee_db_example}); in this case, each party forms a clique of the conflict graph.

Any algorithm for finding (maximal) cliques can be used for this optimization. We used a simple greedy approach in our experiments, but one can deploy here efficient algorithms that guarantee a coverage of the edges by the minimum number of cliques~\cite{hevia2023solvingedgecliquecover}.

Beyond $q=2$, we use the known concept of a \emph{conflict hypergraph} instead of a conflict graph, where the hyperedges are the conflicting candidate sets. Assume that each conflicting set $B[\alpha]$ consists of $q$ candidates, that is, no equality between candidates is allowed. Then the conflict hypergraph is $q$-uniform 
and we contract a set of MIP constraints by eliminating a \emph{hyperclique}: a set $U$ of vertices where every subset of size $q$ is a hyperedge. 
The constraint is then $ \sum_{c_j\in U} < q$.

%% file: section_experimental_evaluation.tex
\input{12charts}

\section{Experimental Evaluation}
We now describe our implementation of the framework and report on the experiments conducted. The goal is to explore the feasibility of the framework via the MIP implementation over realistic data, as well as the effect of the optimizations and of the various 
parameters of the problem.

The implementation\footnote{Code available at \url{https://github.com/Roi-Yona/abcc.git}.}
is programmed in Python3 with SQLite3 as the database engine and Gurobi~\cite{gurobi} as the MIP solver.\footnote{We experimented with several other solvers and found the performance best for Gurobi~11.0.1.}
All experiments were conducted on a machine with 512GB RAM and
64 Intel(R) Xeon(R) Gold 6130 2.10GHz CPUs with 16 cores running Ubuntu 20.04.6 LTS. Due to the number of measurements, each number describes a single run. Note that we consistently use a logarithmic scale on the y-axis.

\subsection{Datasets and Problem Instances} 
We used datasets from three domains: political elections, hotels, and movies. Next, we describe the approval profile, external contexts, and integrity constraints. For each dataset, we imposed one DC and one TGD, 
which we describe here in natural language; formal phrasings are in the Appendix.

\newcommand{\mypara}[1]{\medskip\par\noindent\underline{#1}.}

\mypara{Glasgow City Council elections (2007)}\footnote{\url{https://preflib.simonrey.fr/dataset/00008}}\, This dataset has the results of the 2007 elections of the council of Glasgow, separated by wards (divisions). There are 21 wards with pairwise-disjoint groups of voters and candidates. Each candidate group consists of 8 to 13 people, summing up to 208 candidates. Each ward consists of 5,199 to 12,744 voters, with a total of 188,376 voters. To obtain an ABC instance, we take the union of the wards and establish one set $C$ of candidates and one set $V$ of voters. Each voter $v$ ranks the candidates, and we selected the top three candidates as the approval set $A(v)$. 
As external context, the database $D$ has the following relations: $\rel{Ward}(c,w)$ asserts that candidate $c$ is associated with ward $w$, and $\rel{Party}(c,p)$, from Wikipedia, asserts that  $c$ belongs to party $p$. 
The goal is to elect a committee of candidates. The TGD states that there is at least one committee member from each ward (which is enforced in practice by considering each ward independently). The DC states that no three members belong to the same party.

\mypara{Trip Advisor}\footnote{\url{https://preflib.simonrey.fr/dataset/00040}}\, This dataset contains reviews of 1,851 hotels across the world, scraped from Trip Advisor. Each user ranks hotels on a scale of 1 to 5 (best). The candidate set $C$ consists of the hotels, the voters are the users who ranked more than one hotel, and the approval set $A(v)$ consists of the hotels that $v$ ranked 5. There are 14,137 voters with a nonempty approval set. 
The database $D$ has two relations. 
$\rel{Location}(c,t,p)$ specifies, for each hotel $c$, the city $t$ and country $p$ of $c$. The
relation $\rel{Price}(c,r)$ gives a price range $r$ for each hotel $c$: \val{low}, \val{mid} or \val{high} (derived from the $1/3$ and $2/3$ quantiles). The goal is to select a set of hotels (e.g., to pursue special-price contracts). The DC states that no two hotels have the same city, country, and price range.
The TGD states that there is at least one low-price hotel in each selected location, for selected city-country combinations. For that, we added a relation $\rel{Selected}(t,p)$ with 6 locations. 
  
\mypara{Kaggle's Movies Dataset}\footnote{\footnotesize\url{www.kaggle.com/datasets/rounakbanik/the-movies-dataset}} This dataset integrates data from TMDB and GroupLens, and has movies and user ratings. Ratings are fractional numbers between 1 and 5. The candidate set $C$ consists of 100 movies (first in the dataset list), voters are the reviewing users, and the approval set $A(v)$ of each voter $v$ contains the movies that $v$ ranked above 4. There are 107,733 voters with a nonempty approval set.
The database $D$ has two relations. $\rel{MovieGenre}(c,g)$ specifies, for each movie $c$, the genres $g$ of $c$.  
$\rel{Language}(c,l)$ gives the original language $l$ of each movie $c$. 
The goal is to select a set of movies (say, to show at a social convention). The DC states that no three movies have the same genre. The TGD states that there is at least one movie in English, French, and Spanish. 
\Cref{table:movies} illustrates how different constraints may yield different winning committee, when $k=5$.

\begin{table*}[t]
\caption{Winning committees in the Movies Dataset}
\label{table:movies}
\renewcommand{\arraystretch}{1.0}
\centering
\scalebox{1.0}{
\begin{tabular}{l|l|l|l}
\hline
\textbf{No constraints}   & \textbf{DC}  & \textbf{TGD}  & \textbf{DC+TGD} \\\hline
Judgment Night  & Judgment Night  & Judgment Night & Judgment Night  \\
The Dark & {\color{blue}Endless Summer} & The Dark & {\color{blue}Unforgiven} \\
2001: A Space Odyssey& {\color{blue}Back to the Future} & {\color{blue}Bad Education} &  {\color{blue}Land Without Bread} \\
3 Colours: Red& 3 Colours: Red & 3 Colours: Red & 3 Colours: Red\\
Scarface & {\color{blue}Dracula} &  Scarface & {\color{blue}Back to the Future} \\ \hline
\end{tabular}}
\end{table*}

\subsection{Experiments and Results}
The experimental results are depicted in Figures~\ref{fig:main} and \ref{fig:commitee_size}. We discuss each experiment and the corresponding results. 

The default configuration is as follows.
\begin{itemize}
\item The scoring rule $f$ is PAV. 
\item The set $\Gamma$ of constraints includes the TGD and the DC relevant to each dataset.
\item The MIP construction uses all three optimizations.
\item The committee size $k$ is $10$, except for Glasgow where $k$ is the number of wards; we partitioned the 21 wards into 7 groups of 3, and each tick (x-axis category) adds a group to the experiment (see Figures~\ref{fig:dctgds:glasgow} and~\ref{fig:rules:glasgow}).
\end{itemize}

\paragraph{Effect of the optimizations.}\label{effect_of_the_optimizations_paragraph}
Figures~\ref{fig:opt:const:glasgow},  \ref{fig:opt:const:trip} and \ref{fig:opt:const:movies} show the number of MIP constraints under the different optimizations we described in the previous section: Grouping similar voters (G), Pruning infeasible scores (P), and Contracting DC constraints (C). The horizontal axis is the number of voters. We can see that each of the optimizations reduces the MIP, while the combination of the three gives a reduction of 90\% for Trip Advisor (87.04\%) and Movies (86.81\%). In the case of Glasgow, the unoptimized implementation
reached a timeout before 90,000 voters.
In the Appendix, we give a corresponding chart for the number of variables in the MIP. The trends are similar (except for contraction that does not change the number of variables).

Figures~\ref{fig:opt:time:glasgow},  \ref{fig:opt:time:trip} and \ref{fig:opt:time:movies} show corresponding experiments, except that now we measure the total running time instead of the number of constraints. We can see that the behavior is quite similar in the total 
effect compared to the vanilla implementation, but the 
effect
of each individual optimization varies. For example, the addition of Pruning has little 
effect
(if any) in the Glasgow and the Movies datasets, but Contraction is consistently accelerating the computation by an order of magnitude. On the other hand, Contraction has less 
effect in the case of Trip Advisor. 

\paragraph{Effect of the database constraints.}
Recall that we have one DC and one TGD for each dataset. In the experiments depicted in Figures~\ref{fig:dctgds:glasgow}, \ref{fig:dctgds:trip} and \ref{fig:dctgds:movies}, we applied our solution to find a winning committee under four configurations: (1) no constraints, (2) only the DC, (3) only the TGD, and (4) both the DC and the TGD. Each bar shows the time spent on the MIP construction (bottom, shaded by line patterns) and the MIP solving (top). The x-axis is the number of voters. 

From these experiments, we draw several insights. First, the database constraints have no significant effect on the running time. An exception is Glasgow, where the DC leads to a slowdown of up to an order of magnitude ($14.04\times$). Interestingly, once the TGD is added to the DC, the running time drops down to the constraint-free configuration. Hence, the addition of constraints can lead to a reduction in the solver's running time. Second, the construction time takes around $6\%$ to $1/3$ of the total computation in the case of Glasgow and Trip Advisor. (Note the logarithmic scale.) 
In the Movies Dataset, the construction time is negligible.

\input{committee-size}

\paragraph{Effect of the scoring rule.}
In this experiment, we varied the scoring rule. Figures~\ref{fig:rules:glasgow}, \ref{fig:rules:trip} and \ref{fig:dctgds:movies} show the results. We experimented with five scoring rules: (1) Approval Voting (AV), (2) Proportional Approval Voting (PAV), (3) Chamberlin-Courant (CC), (4) Satisfaction Approval
Voting (SAV), and (5) 2-truncated Approval Voting (2AV). The first four are defined in the Preliminaries. 2AV is similar to AV, except that a candidate can contribute at most two to the committee: $w(x)=\min(2,x)$.
Note the construction time is not affected by the scoring rule (as expected). Interestingly, the computation time of CC and 2AV is consistently faster than the rest, typically by 5 times (e.g., Glasgow) to 100 times (e.g., Movie Dataset). It appears that the MIP solver utilizes the fact that $f(k',|A(v_i)|)$ is the same for all $k'\geq 2$.

\paragraph{Effect of the committee size.}
\Cref{fig:commitee_size} shows the results where we varied the size $k$ of the committee. We did so only on Trip Advisor and Movies Dataset since the committee size in Glasgow is the number of wards. In the Movies Dataset, we excluded the DC to allow for the existence of a feasible committee. In Trip Advisor, the running time can be greatly affected by $k$, while the difference is smaller between $k=40$ and $k=50$. In the Movies Dataset, there is a consistent peak between $k=20$ and $k=30$. This behavior may be due to the ratio between $k$ and $|C|$, but not due to the type of constraints, as we observed similar phenomena with different combinations of constraints (see the Appendix).

\subsection{Discussion}
We learn several lessons from the experiments. First, the combined effect of the three optimizations is significant, even though each optimization may contribute differently in different domains. Second, the database constraints incur a cost. However, if we start with one DC, adding a TGD does not necessarily increase the execution cost; in fact, the cost usually drops. In the Appendix, we show that this holds even when adding two more TGDs.
Third, the MIP solver is highly sensitive to the scoring rule, while the important factor for performance seems to be the number of different score values per voter. 
Finally, the size of the committee has a significant effect on the running time; yet, with constraints the behavior is not monotonic as one might expect.

%% file: 12charts.tex
\begin{figure*}[th]
\newlength{\wdt}
\setlength{\wdt}{0.238\textwidth}
\small
\centering
    Glasgow City Council
    \vskip-0.1em
    \subfigure[]{\label{fig:opt:const:glasgow}   
        \includegraphics[width=\wdt]{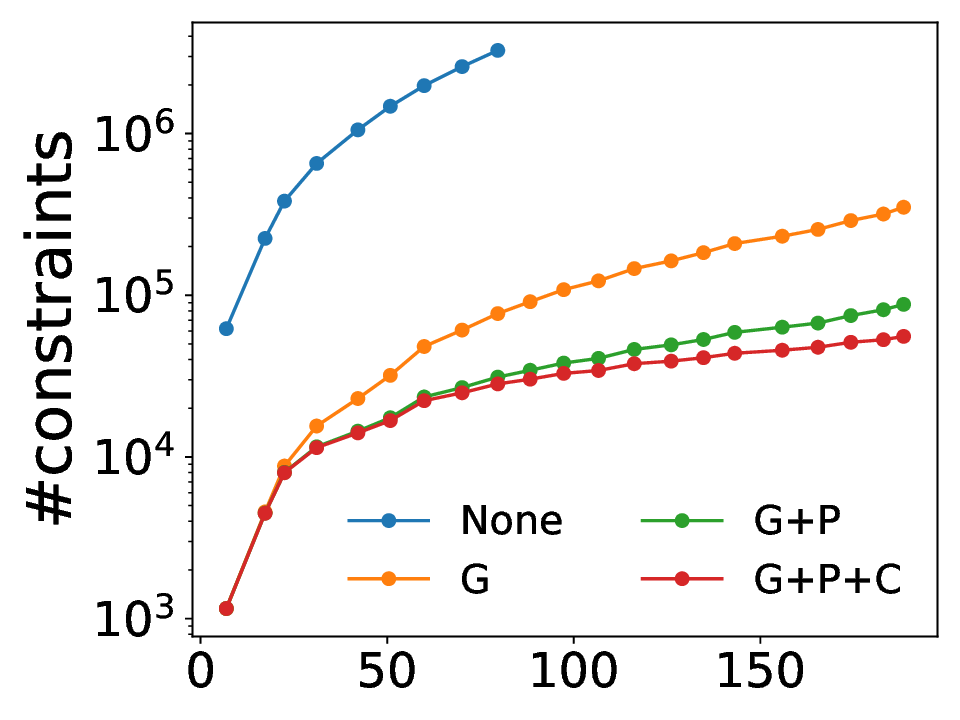}}
    \subfigure[]{\label{fig:opt:time:glasgow} 
        \includegraphics[width=\wdt]{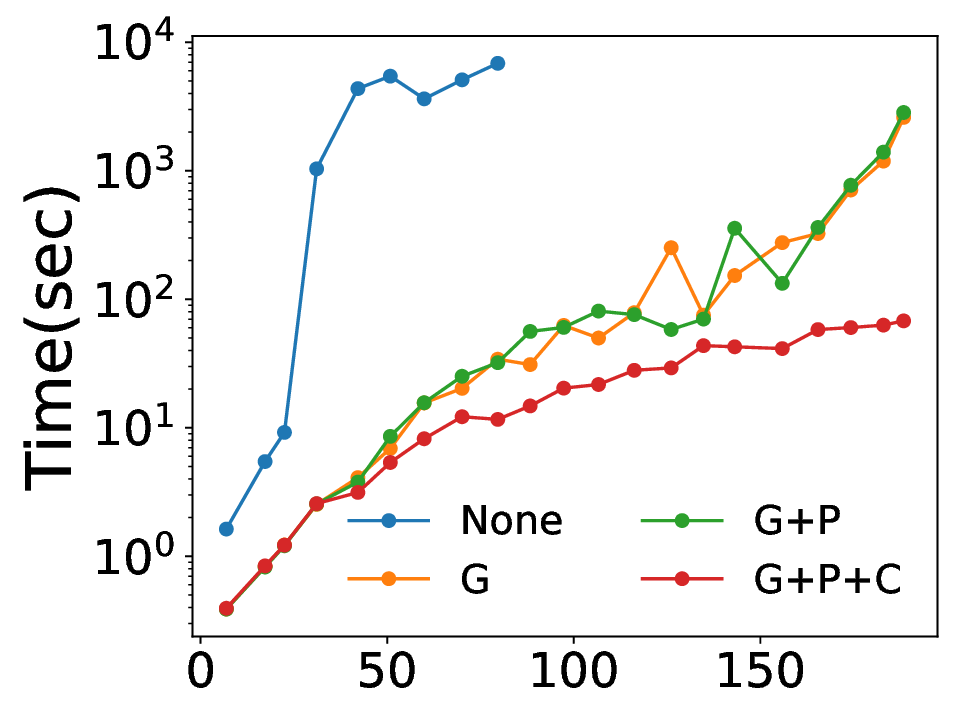}}
    \subfigure[]{\label{fig:dctgds:glasgow} 
        \includegraphics[width=\wdt]{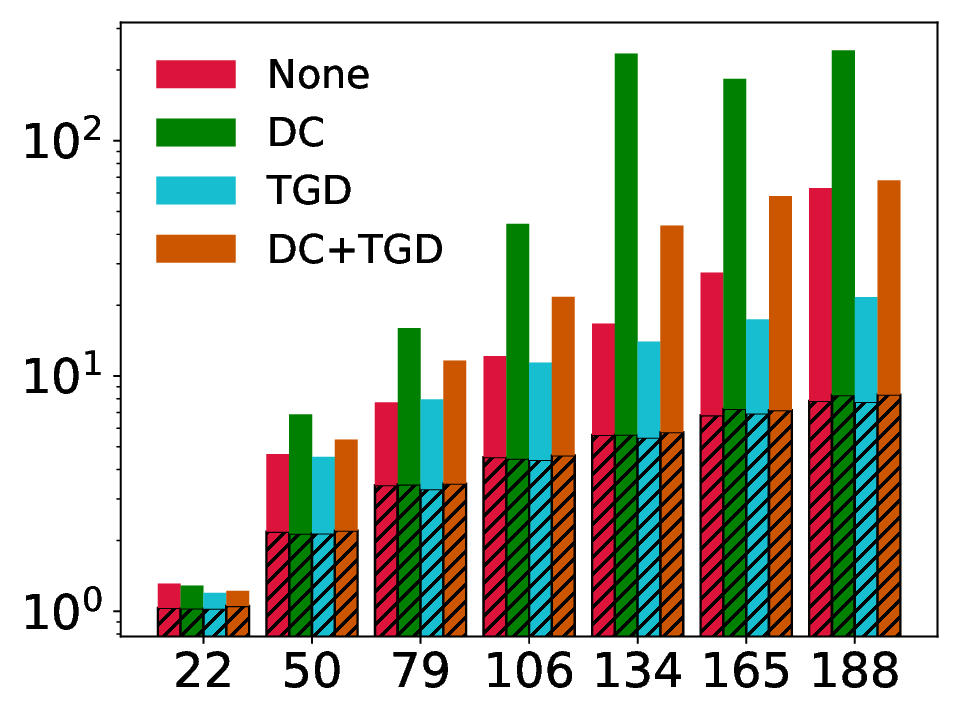}}
    \subfigure[]{\label{fig:rules:glasgow} 
        \includegraphics[width=\wdt]{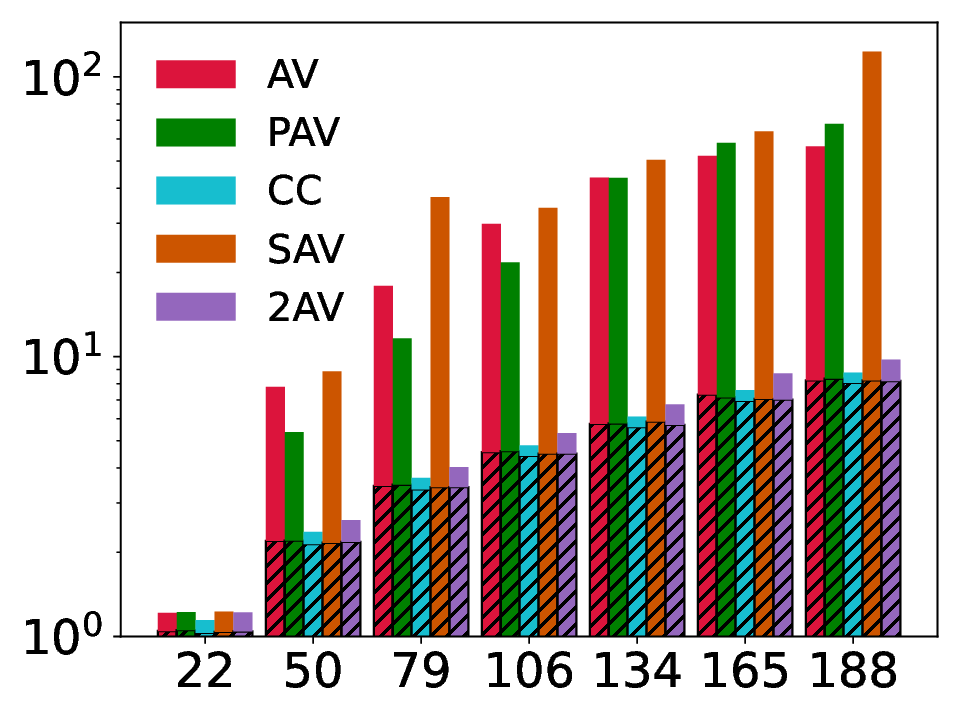}}\\
    Trip Advisor
    \vskip-0.1em
    \subfigure[]{\label{fig:opt:const:trip}
        \includegraphics[width=\wdt]{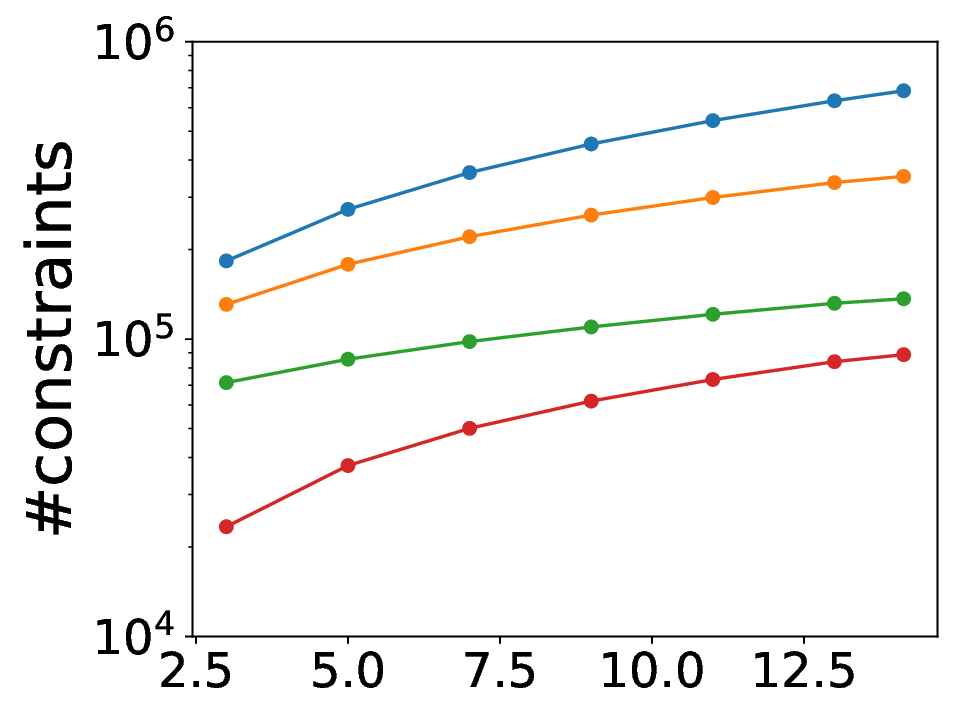}
       }\!\!\!
   \subfigure[]{\label{fig:opt:time:trip}
        \includegraphics[width=\wdt]{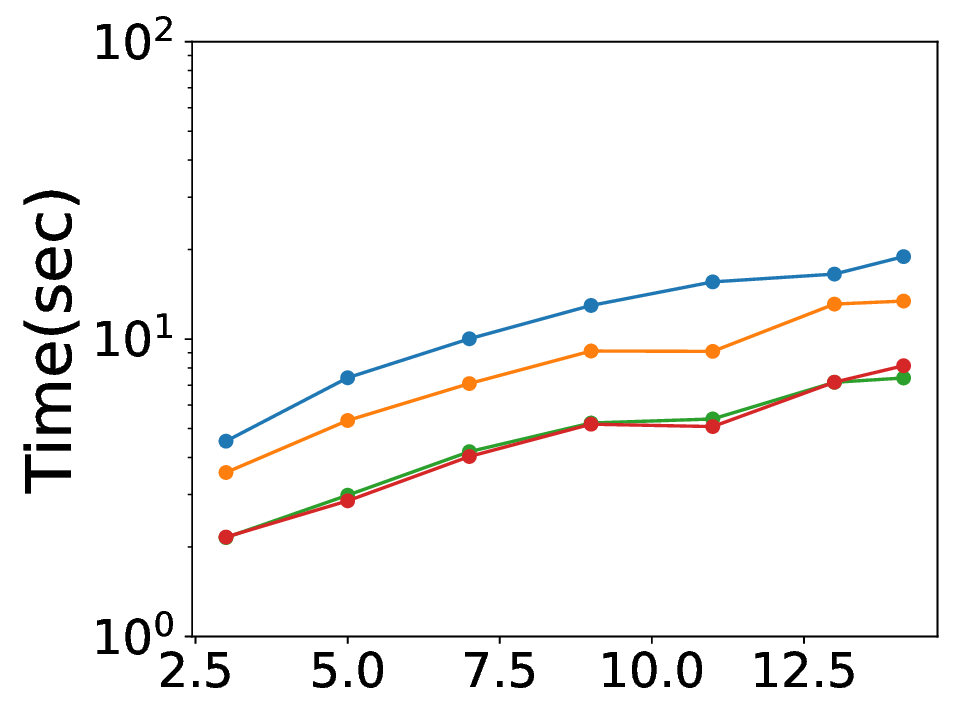}
    }\!
    \subfigure[]{\label{fig:dctgds:trip} 
        \includegraphics[width=\wdt]{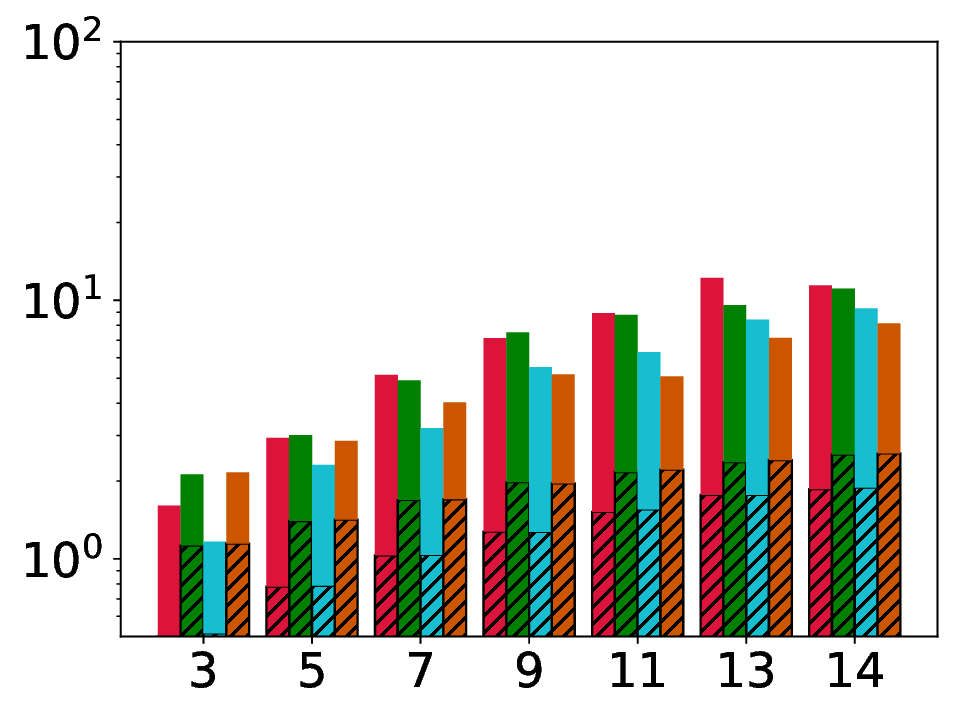}
    }\!\!  
    \subfigure[]{\label{fig:rules:trip} 
        \includegraphics[width=\wdt]{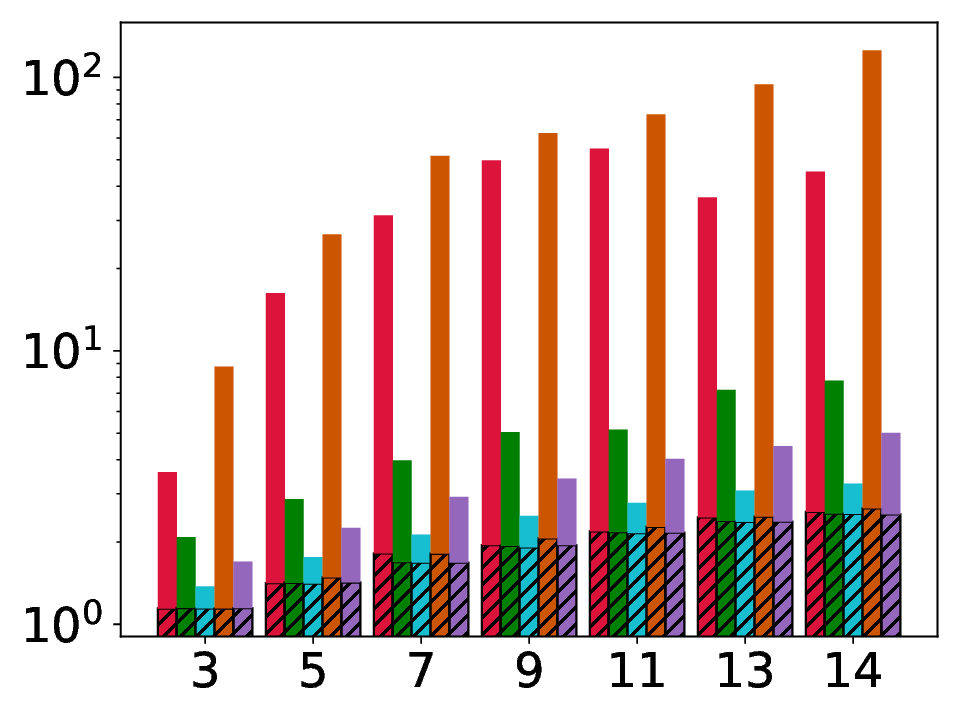}
    }\\
    Movies Dataset
    \vskip-0.1em
    \subfigure[]{\label{fig:opt:const:movies}
        \includegraphics[width=\wdt]{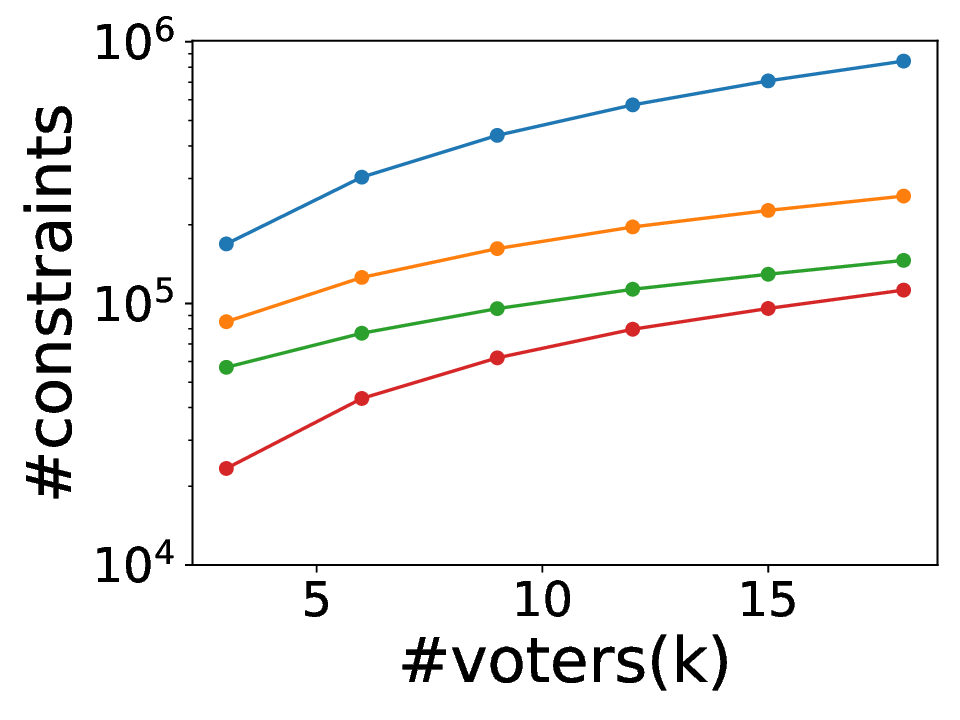}
     }\!\!\!
    \subfigure[]{\label{fig:opt:time:movies}
        \includegraphics[width=\wdt]{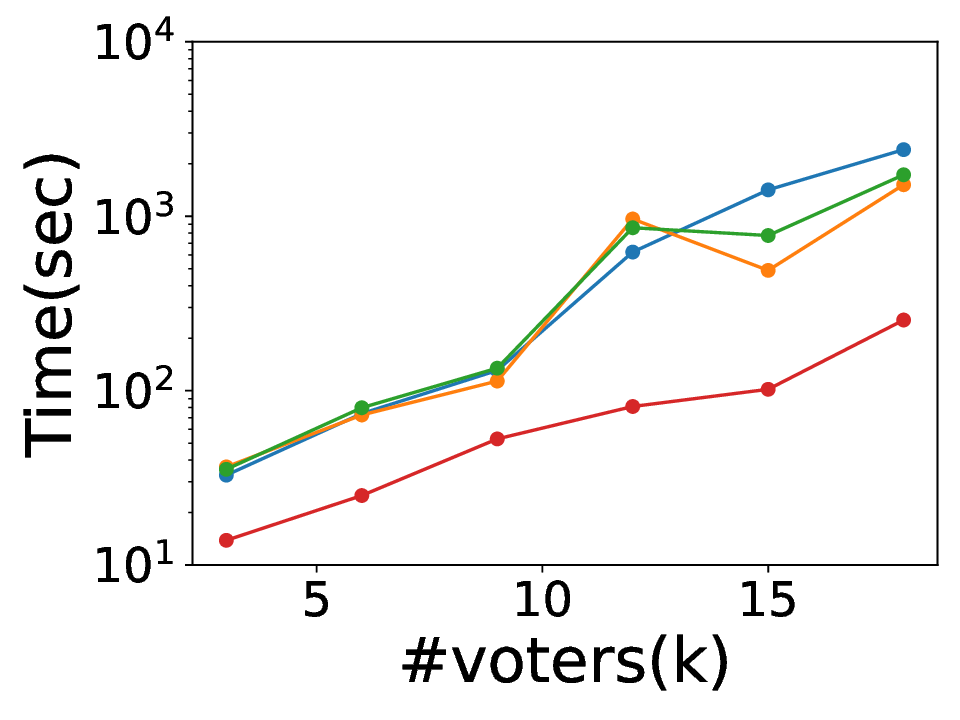}
     }\!\!
    \subfigure[]{\label{fig:dctgds:movies} 
        \centering
        \includegraphics[width=\wdt]{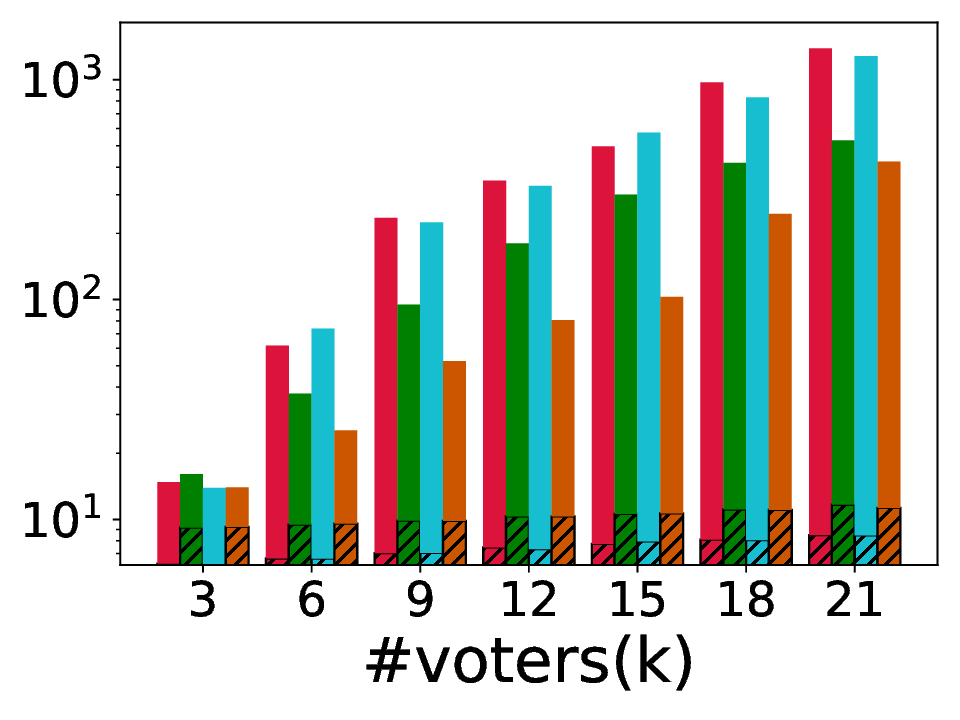}
         }\!\!\!
    \subfigure[]{\label{fig:rules:movies} 
        \centering
        \includegraphics[width=\wdt]{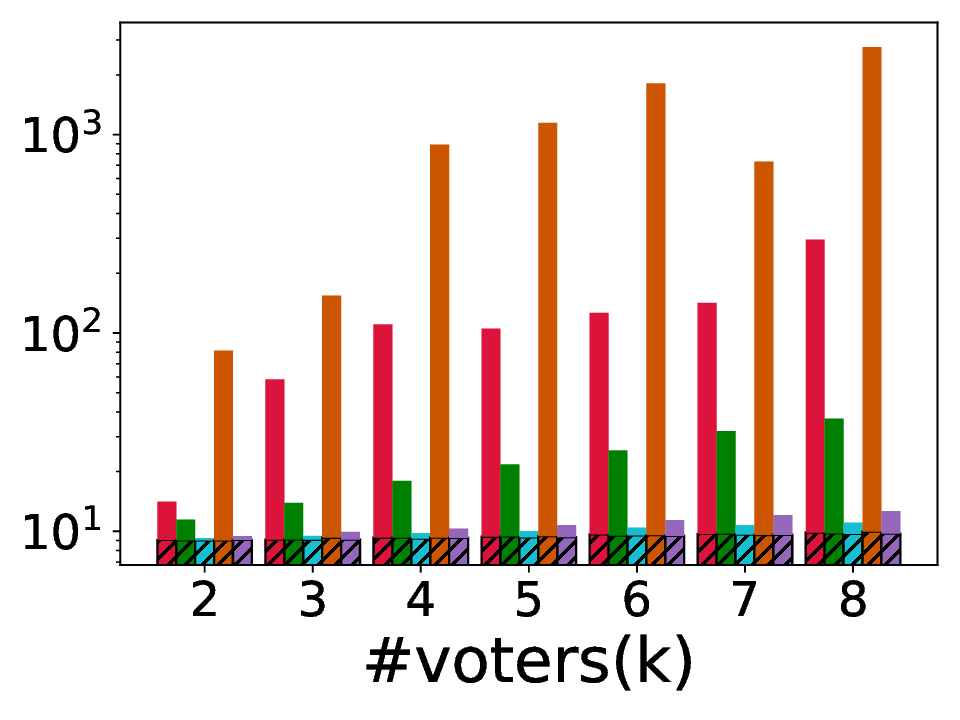}
        }
    \caption{From left to right: impact of the optimizations on the number of MIP constraints (first) and runtime (second), impact of varying constraints (third) and voting rules (fourth) on the runtime. G, P, and C refer to the optimizations:
    Grouping similar voters, Pruning infeasible scores, and Contracting DC constraints.
    }
    \label{fig:main}
\end{figure*}

%% file: committee-size.tex
\begin{figure}[t] 
    \centering
\subfigure[Trip Advisor ($|C|=1,845$)]{\includegraphics[width=0.22\textwidth]{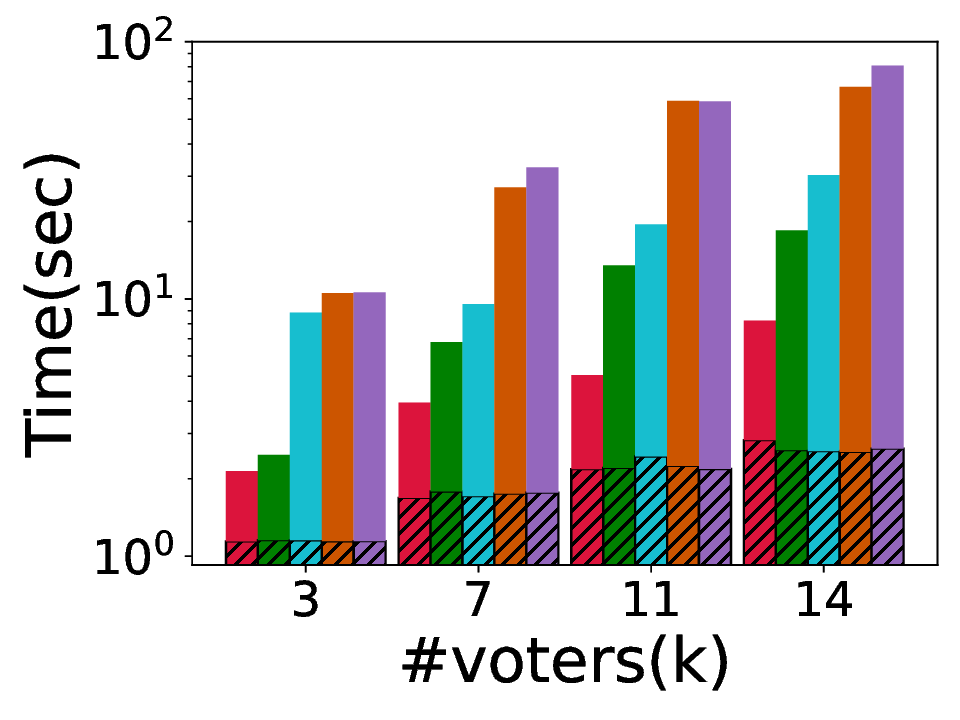}       \label{fig:trip_advisor_different_committee_size}
        }
     \subfigure[Movies ($|C|=100$)]{
        \includegraphics[width=0.22\textwidth]{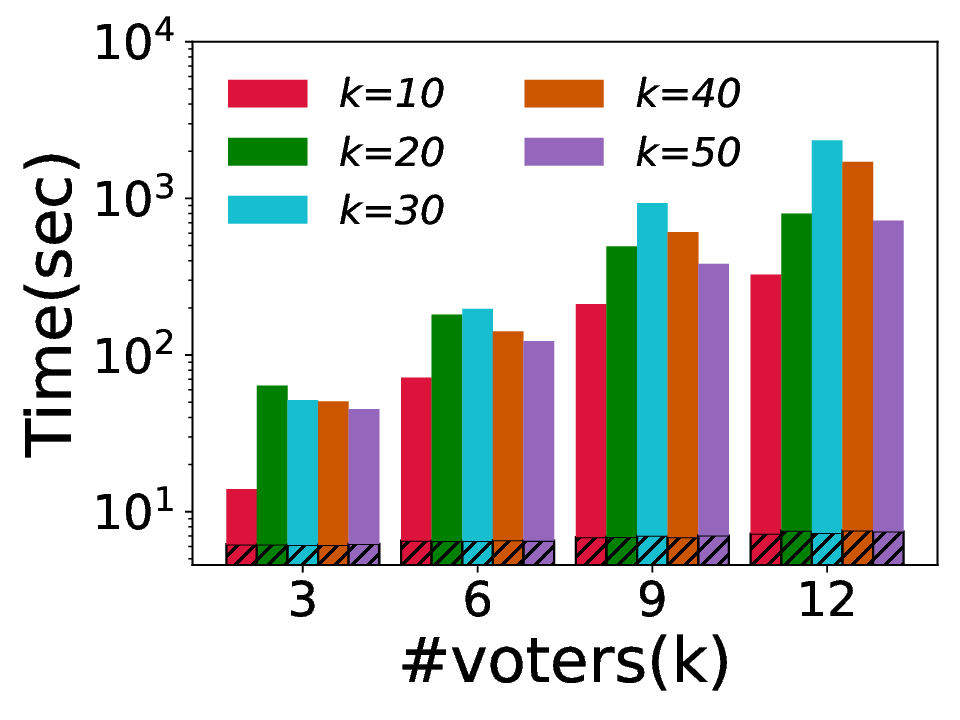}       
        \label{fig:movies_different_committee_size}
        }
        \vskip-1em
    \caption{Computing time for varying committee size $k$.}
    \vskip-1em
    \label{fig:commitee_size}
\end{figure}

%% file: main.bbl

%% file: appendix.tex
\def\cin#1{c_{#1}^{\mathsf{in}}}
\def\cout#1{c_{#1}^{\mathsf{o}}}
\def\cst#1{\mathrm{cost}(c_{#1})}

\newcommand{\repthm}[2]{
\smallskip
\noindent \textbf{Theorem #1.}\, {\em #2}
\smallskip
}

\appendix

\section{Proof of Theorem~\ref{thm:tgd}}

\repthm{\ref{thm:tgd}}{\thmtgdtxt}

We prove \Cref{thm:tgd} by considering several cases that we handle in the following lemmas.
For the problem of determining whether a legal committee exists,
membership in NP is straightforward; the legal committee itself can be a witness of a ``yes'' instance. Hence, the proofs of NP-completeness will focus on hardness.

\begin{lemma}\label{lemma:t=1NPC}
In the case of $t=1$, and assuming no key constraints, it is NP-complete to determine whether a legal committee exists. 
\end{lemma}
\begin{proof}
We prove hardness by a reduction from the Set Cover problem. In this problem, we are given as input a set  $U=\{1,..,n\}$ of elements, a collection $Q=\{s_1,...,s_m\}$  of subsets of $U$, and an integer $k$. The goal is to determine whether $Q$ contains $k$ subsets whose union is equal to $U$. 

Given the input $(U,Q,k)$, we the set $C=\{c_1,...,c_m\}$ of candidates, and a database $D$ over $\scs$ with $R_1^D=\{(1),...,(n)\}$ and $S_1^D=\{(i,c_j)\mid i\in s_j\}$. According to this construction, the TGD states that every $i\in U$ is represented in $\comrel$ by a candidate $c_j$ such that $i\in s_j$. Hence, a legal committee of size $k$ exists if and only if there is a cover of $U$ by $k$ sets from $Q$. 
\end{proof}

The next lemma generalizes \Cref{lemma:t=1NPC}.
\begin{lemma}\label{lemma:$t>=1$,NP-C}
For all $t\ge1$, it is NP-complete to determine whether a legal committee exists is NP-complete, as long as at least one $S_\ell^D$  has no key constraints. 
\end{lemma}
\begin{proof}
We reduce from the case of $t=1$ to $t'>1$. Given the input $(C,k,D)$ and $D$ for $t=1$, we construct the input $(C',k',D')$ for $t'>1$, as follows. Let $c$ be a new candidate not in $C$. Then $C'=C\cup\set{c}$ and $k'=k+1$. The relations of $D'$ for $R_1$ and $S_1$ are those of $D$, that is, 
$R_1^{D'}=R_1^D$ and $S_1^{D'}=S_1^D$.   
For $\ell>1$, we construct single-tuple relations $R_\ell^{D'}=\set{(m+1)}$ and $S_\ell^{D'}=\set{(m+1,c)}$. This construction ensures that a set $B$ of candidates is a legal committee for $(C,k,D)$ if and only if $B\cup\set{c}$ is a legal committee for $(C',k',D')$. Hence, the existence of a legal committee is equivalent between the cases, thus the correctness of the reduction.
\end{proof}

The following lemma states tractability for $t=1$ in the presence of a key constraint.
\begin{lemma}\label{lemma:$t=1$, key constraint, P}
Consider the case of $t=1$ where the first attribute in $S_1^D$ is a key constraint. For AV, a winning committee can be found in polynomial time.
\end{lemma}
\begin{proof}
In this case, we can solve the problem using a simple greedy algorithm. Since the candidate is a key in $S_1$, we can partition the candidate set $C$ according to the associated value in $S_1$. Hence, a partition $C_a$ is associated with a value $a$, and it includes all candidates $c$ such that $(c,a)\in S_1^D$. Then, for each $(a)\in R_1^D$ we select a candidate $c\in C_a$ with a maximal number of approvals (since the voting rule is AV). If $B$ already contains more than $k$ candidates at this point, then no legal committee exists. Otherwise, we complete $B$ to a committee of size $k$ by adding $k-|B|$ remaining candidates with a maximal number of approvals.
\end{proof}

Finally, the next lemma proves tractability in the case of $t=2$ in the presence of a key constraint in $S_1$ and $S_2$.
\begin{lemma}
\label{lemma:t=2P}
In the case of $t=2$ where the first attribute in $S_1$ and $S_2$ is a key constraint, we can find a winning committee under AV in a polynomial time. 
\end{lemma}
\begin{proof}
We solve the problem by a reduction to the Minimum Cost Maximum Flow (MCMF) problem.
The input to MCMF is a directed graph $(U,E)$ with distinguished vertices $s'$ and $t'$ in $U$, where every edge $e\in E$ has a capacity $\kappa_e$ and a cost $\rho_e$. A flow is a function $f:E\rightarrow\mathbb{R}$ such that $0\le f(e)\le\kappa_e$ for every $e\in E$; moreover, for every vertex $u\in U$, except for $s'$ and $t'$, it holds that $\sum_{e\in \mathrm{in}(u)} f(e)=\sum_{e\in\mathrm{out}(u)} f(e)$. 
Here, $\mathrm{in}(u)$ and $\mathrm{out}(u)$ are the sets of incoming and outgoing edges of $u$, respectively. By a maximum flow, we mean a flow $f$ that maximizes $\sum_{(s',u)\in E} f(s',u)$. A \emph{minimum cost maximum flow} is a maximum flow $f$ with a minimal cost, where the cost of a flow is  $\sum_{e\in E}f(e)\cdot \rho_e$. The flow $f$ is \emph{integral} if all values $f(e)$ are integers. It is known that whenever the capacities are natural numbers, an integral minimum cost maximum flow exists and, moreover, can be found in polynomial time~\cite[Chapter 9]{DBLP:books/daglib/0069809}.

First, we note that due to the key constraint, there is no legal committee if $|R_1^D|>k$ or $|R_2^D|>k$ (because of the key constraint every candidate $c$ has at most one $(a)\in R_1^D$ such that $(c,a)\in R_1^D$ and the same goes for $R_2^D$). 

Given the input $(C,k,D)$ to our problem, we define the network $(U,E)$ as follows. (See illustration in \Cref{fig:mcmf}.)
The vertex set $U$ consists of the following vertices.
\begin{itemize}
\item The source $s'$ and the sink $t'$;
\item Vertices $\cin j$ and $\cout j$ for every candidate $c_j\in C$;
\item A vertex $u_a$ for every $(a)\in R_1^D$ and a vertex $w_b$ for every $(b)\in R_2^D$;
\item A vertex $u'_i$ for all $i=1,\dots, k-|R^D_1|$ and a vertex 
$w'_i$ for all $i=1,\dots, k-|R^D_2|$. (Recall our assumption that $|R_1^D|\le k$ and $|R_2^D|\le k$.)
\end{itemize}


The edge set $E$ consists of the following edges, all with a unit capacity.
\begin{itemize}
\item A edge from $s'$ to every $u_a$ and $u'_i$;
\item An edge from $u_a$ to $\cin j$ for all tuples $(a)\in R_1^D$ and $(c_j,a)\in S_1^D$;
\item An edge from $u'_i$ to $\cin j$ for all 
$i=1,\dots, k-|R^D_1|$ and candidates $c_j\in C$;
\item An edge from $\cin j$ to $\cout j$ for all $c_j\in C$;
\item An edge from $\cout j$ to $w_b$ for all tuples $(b)\in R_2^D$ and $(c_j,b)\in S_2^D$;
\item An edge from $\cout j$ to $w'_i$ for all candidates $c_j\in C$ and $i=1,\dots, k-|R^D_2|$;
\item A edge from every $w_b$ and $w'_i$ to $t'$.
\end{itemize}


The cost of every edge is $0$, except for the edges $e$ from $\cin j$ and $\cout j$ where the cost of $e$ is $|V|-|V_j|$ where $V$ is the set of voters (as usual) and $V_j$ is the set of voters who approves $v_j$; in \Cref{fig:mcmf} we denote this number by $\cst j$.


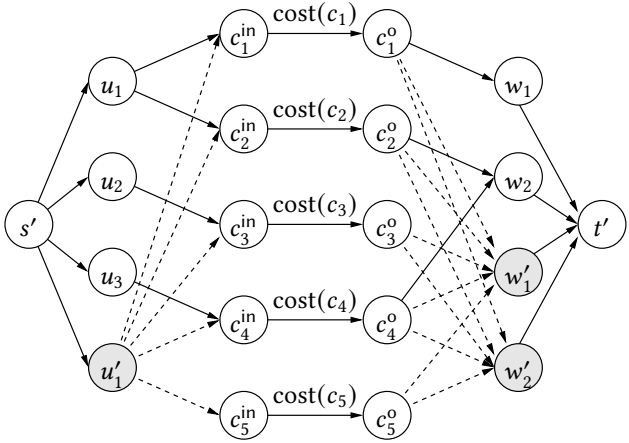
\begin{figure}
  \centering
  \input{mcmf.pspdftex}
   \caption{Network built in the reduction of \Cref{lemma:t=2P}.}
   \label{fig:mcmf}
\end{figure}

For illustration, in \Cref{fig:mcmf} you can see the network for the following input:
\begin{itemize}
    \item $C=\set{c_1,\dots,c_5}$
    \item $k=4$
    \item $R_1^D=\set{(1),(2),(3)}$
    \item $R_2^D=\set{(1),(2)}$
    \item $S_1^D=\set{(c_1,1),(c_2,1),(c_3,2),(c_4,3)}$
    \item $S_2^D=\set{(c_1,1),(c_2,2),(c_4,2)}$
\end{itemize}

If there is a legal committee $B$, then we can transfer a flow of $k$ units from $s'$ to $t'$. Indeed, for each $u_a$ we can select a destination $c_j\in B$ such that $(c_j,a)\in S_1^D$, and complete the $k$ units using edges from the $u'_i$. Note that we do not share the same $c_j$ for two $u_a$s due to the key constraint. We can similarly continue the flow through the $w_b$ and $w'_i$, and finally to $t'$. Similarly, a flow of size $k$ can be transformed into a legal committee $B$ by taking the candidates $c_j$ such that there is flow through $\cout j$. 

Hence, there is a correspondence between the legal committees and the (maximum) flows of amount $k$. From our definition of the costs of edges from $\cin j$ to $\cout j$, we conclude that a minimal-cost flow of amount $k$ selects a committee $B$ with a minimal cost that, due to the choice of cost, maximizes the sum $\sum_{c_j\in B}|V_j|$, that is, the AV score. Hence, a solution to MCMF yields a winning committee, as required. 
\end{proof}

Now, we prove the hardness in the case of $t=3$, and the presence of key constraints.
\begin{lemma}\label{lemma:t=3NPC}
Consider the case of $t=3$ where the first attribute in the relations $S_1^D$ and $S_2^D$, and $S_3^D$ are key constraints. It is NP-complete to determine whether a legal committee exists. 
\end{lemma}
\begin{proof}
To show hardness, we reduce from the \emph{3-Dimensional Matching} (3-DM) problem, which is the following. Let $X$, and $Y$ and $Z$ be finite sets, and let $T$ be a subset of $X\times Y\times Z$ consists of triplets $(x,y,z)$ such that $x\in X$ and $y\in Y$ and $z\in Z$. $M\subseteq T$ is a 3-dimensional matching if the following holds: for every two distinct triplets $(x_1,y_1,z_1)\in M$ and $(x_2,y_2,z_2)\in M$, we get that $x_1\neq x_2$, that $y_1\neq y_2$, and that $z_1\neq z_2$. Finding $n$ size 3-dimensional matching in a given hypergraph, when $|X|=|Y|=|Z|=n$ is NP-Complete \cite{Karp1972}.

We are given the triple $(X,Y,Z)$ such that $|X|=|Y|=|Z|=n$. Let $T$ be a subset of $X\times Y\times Z$ consisting of triples $(x,y,z)$ such that $x\in X$, that $y\in Y$, and that $z\in Z$. We enumerate each triple in $T$ starting from $1$ to $|T|$.

We define $C=\set{1,...,|T|}$ the set of candidates, $k=n$ the committee size, and a database $D$ over $\scs$ such that $R_1^D=X$ and $R_2^D=Y$ and $R_3^D=Z$ (where we identify an element an element $a$ with the tuple $(a)$) and for $1\leq i\leq 3$, the relation $S_i^D$ consists of the tuples of $(c,a)$, such that $c\in C$ and $(a)\in R_i^D$ and the fit enumerated triplet to $c$ contains $a$ in the $i$ place.

Each element in a triplet is contained in a different relation $R_i^D$ (corresponding to its position in the triplet). This is due to the $S_i^D$ properties (where a candidate $j\in C$ and $(a)\in R_i^D$ satisfy $(j,a)\in S_i^D$ if and only if $a$ is the $i$ element in the corresponding $j$ triplet), we get that every candidate $c$ represents one triplet in $T$. Also, we satisfy the key constraint in the first attribute of $S_1^D$ and $S_2^D$, and $S_3^D$ (since every candidate $j$ is in the relation $S_i^D$ only with the corresponding element in the $j$ triplet in the $i$ place).
From these properties we get that there is 3-DM from size $n$ if and only if there is a legal committee given the input $(C,k,D)$, and, $t=3$ (with key constraint in the first attribute of $S_1^D$ and $S_2^D$, and $S_3^D$). 
\end{proof}

Lastly, similar to previous proofs, we reduce from the case where $t=3$ to the case where $t\ge 3$.
\begin{lemma}\label{lemma:t>3NPC}
Consider the case of $t\ge3$ where the first attribute is a key constraint for all the relations $S_i^D$, such that $1\leq i\leq t$. It is NP-complete to determine whether a legal committee exists.
\end{lemma}
\begin{proof}
From \Cref{lemma:t=3NPC} we got that in the case of $t=3$, the problem is NP-complete. To show hardness, we reduce the case of $t=3$ to $t'\ge3$.

Given the input $(C,k,D)$ for $t=3$, we define a new database $D'$ for $t'\ge3$ such that for all $1\le i \le 3$ it holds that $R_i^{D'}=R_i^D$ and $S_i^{D'}=S_i^D$, furthermore, for all $3<i\le t'$ it holds that $R_i^{D'}=R_1^D$, and $S_i^{D'}=S_1^D$.

We get that for all $3<i\le t'$ the relations $R_i^{D'}$ and $S_i^{D'}$ have no additional effect on the legality of the committees. Therefore, there is a legal committee given the input $(C,k,D')$ for $t'\ge3$ if and only if there is a legal committee given the input $(C,k,D)$ for $t=3$.
\end{proof}

This completes the proof of \Cref{thm:tgd}.

\section{Proof of Theorem~\ref{thm:dc}}
\repthm{\ref{thm:dc}}{\thmdctxt}

We prove the theorem by considering several cases that we handle in the following lemmas. Again, for the problem of determining whether a legal committee exists, membership in NP is straightforward, and we focus on hardness in the proofs of NP-completeness. 

First, we prove that it is NP-complete to determine whether a legal committee exists. 
\begin{lemma}
Given the first argument of $R^D$ is not necessarily a key constraint, it is NP-complete to determine whether a legal committee exists.
\end{lemma}
\begin{proof}
To show hardness, we reduce from the \emph{Maximal Independent Set} problem to our problem.

In the \emph{Maximum Independent Set} problem we are given an undirected graph $G=(V_G,E_G)$, and the goal is to find the size of the Maximum Independent Set. A Maximum Independent Set is a set of disjoint vertex of the graph with a maximum size. This is a known NP-hard problem \cite{DBLP:books/fm/GareyJ79}.

Given a graph $G=(V_G,E_G)$, such that $V_G=\set{1,..,m}$ is the group of vertex. We define $C=\set{1,\ldots,m}$ the candidate group, and a database $D$ over $\scs$, such that $R^D=\set{(j_1,a_{(j_1,j_2)})\mid (j_1,j_2)\in E_G} \cup \set{(j_2,a_{(j_1,j_2)})\mid (j_1,j_2)\in E_G}$.

Due to the properties of $R^D$, and the DC constraint, we get that every legal committee is an independent set, and vice versa. Therefore, we can find the maximum independent set size by searching for a legal committee, starting from $k=m$ to $k=1$. This concludes the hardness side.
\end{proof}

Next, we prove tractability in the case where the first attribute of $R^D$ is a key constraint.
\begin{lemma}
Given the first attribute in $R^D$ is a key constraint and AV as the voting rule, we can find a winning committee in polynomial time.
\end{lemma}
\begin{proof}
To prove traceability we define the following polynomial algorithm.
First, we note that for each $(a)\in R^D$ there is a group of candidates that are in a so-called conflict, i.e. cannot be in the same committee.

We define $C'=C$. For all $(a)\in R$ we mark the group of candidates that are in conflict due to this element as $U_a$. In AV every candidate denotes a separate score to the committee score (the number of votes this candidate received). We sort $U_a$ based on this score, and remove the highest score candidate from it.
Furthermore, we subtract $C'=C'\setminus U_a$. We repeat this process for all $(a)\in R^D$. From the resulting $C'$ we choose the $k$ highest score candidates (just like in regular AV, without constraints) to be our committee (and return false if $|C'|<k$, i.e. there is no legal committee).

Due to the key constraint, all the groups $U_a$ are disjoint. Therefore, a winning committee that is composed of candidates from these groups contains at most one candidate from each group (only the highest score one). Hence, the algorithm returns a winning committee if exists, and false otherwise.
\end{proof}
This completes the last part in \Cref{thm:dc}.

\section{Additional Experiments}
\begin{figure*}[th]
\small
\centering
   \subfigure[Glasgow City Council]{
        \includegraphics[width=0.23\textwidth]{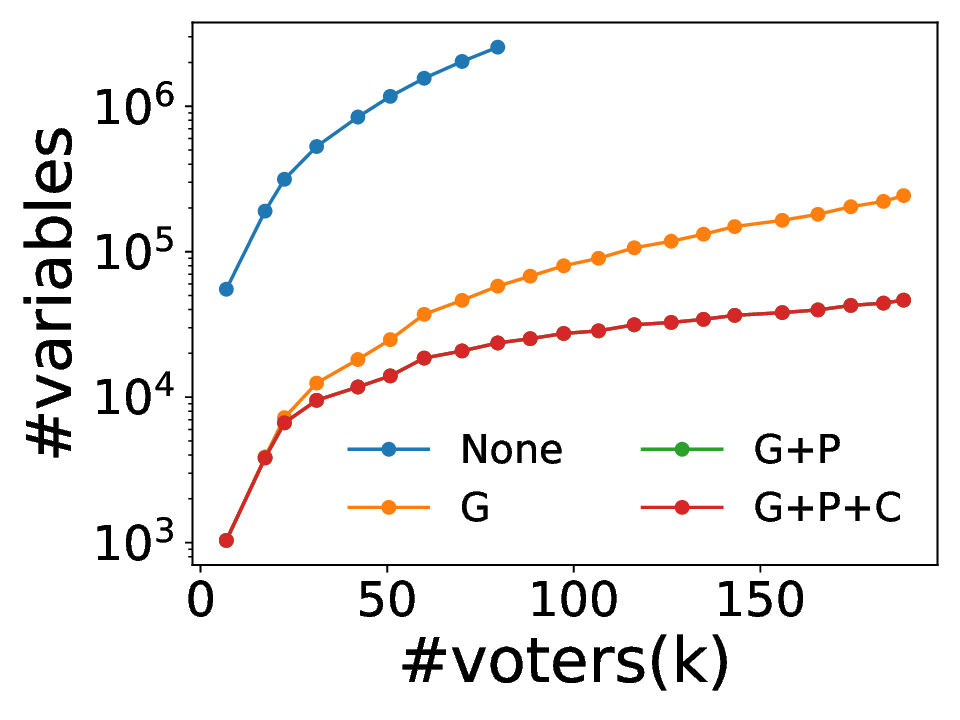}
        \label{fig:glasgow_election_different_optimizations_variables}
    }
    \subfigure[Trip Advisor]{
        \centering
        \includegraphics[width=0.23\textwidth]{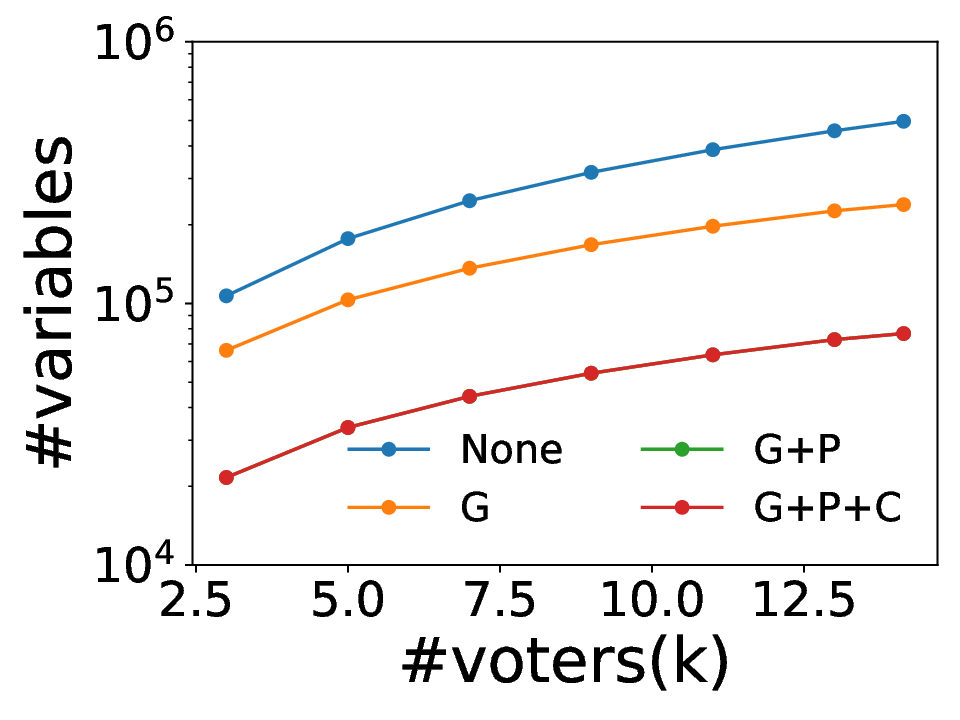}
        \label{fig:trip_advisor_different_optimizations_variables}
        }
    \subfigure[Movies Dataset]{
        \centering
        \includegraphics[width=0.23\textwidth]{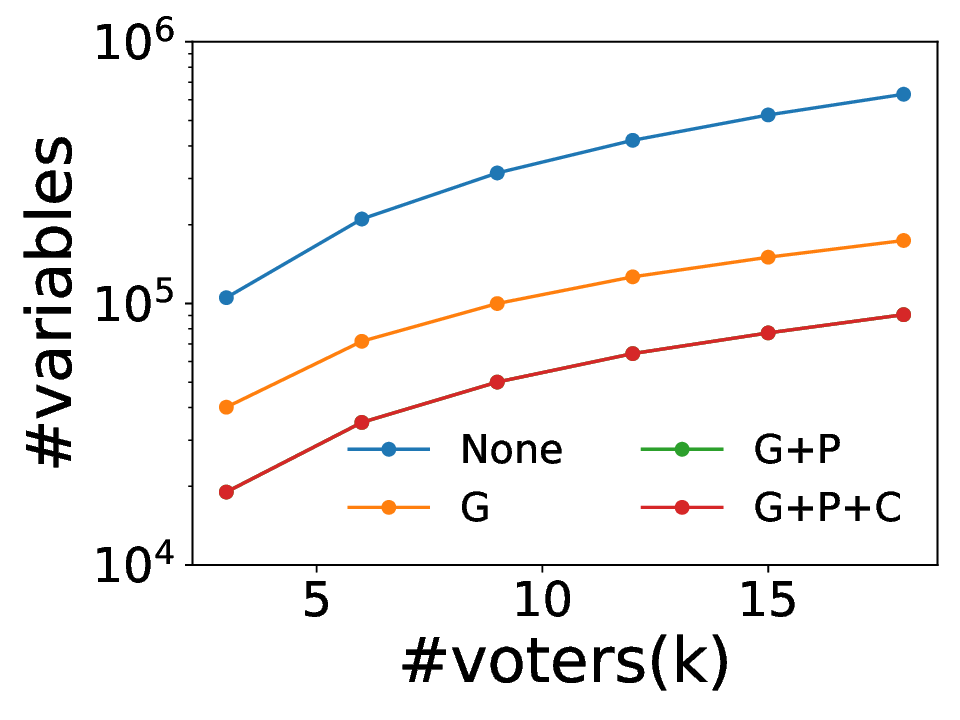}
        \label{fig:movies_different_optimizations_variables}
    }
    \subfigure[Movies Dataset]{
        \centering
        \includegraphics[width=0.23\textwidth]{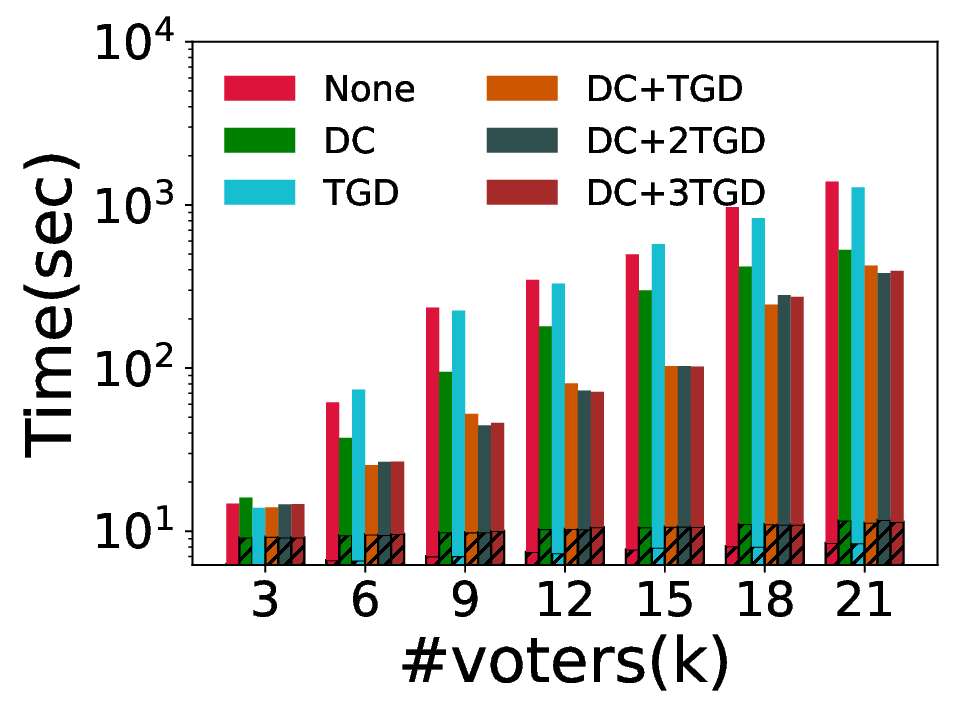}        
         \label{fig:movies_different_constraints_multiple_constraints}
    }
    \caption{The markings G, P, and C refer to the different optimizations - grouping similar voters, pruning infeasible scores, and contracting DC constraints via hypercliques accordingly.}
    \label{fig:different_optimizations_variables}
\end{figure*}
We give here additional experiments to complement those given in the Experimental Evaluation section. The experimental results are displayed in \Cref{fig:glasgow_election_different_optimizations_variables}, \Cref{fig:trip_advisor_different_optimizations_variables}, \Cref{fig:movies_different_optimizations_variables} are the corresponding optimization experiments of the Figures~\ref{fig:opt:const:glasgow},  \ref{fig:opt:const:trip}, \ref{fig:opt:const:movies} and 
Figures~\ref{fig:opt:time:glasgow},  \ref{fig:opt:time:trip}, \ref{fig:opt:time:movies}, except this is the measurement of the number of constraints under the different optimizations. As expected, the reduction of the optimizations on the number of variables is quite similar to the effect on the constraint number with the exception of the last optimization, contracting DC constraints via hypercliques (C). This is expected, because this optimization only contracts inequations, without any effect on the variables.

\begin{figure}[t] 
    \subfigure[Trip Advisor]{
    \centering
     \includegraphics[width=0.222\textwidth]{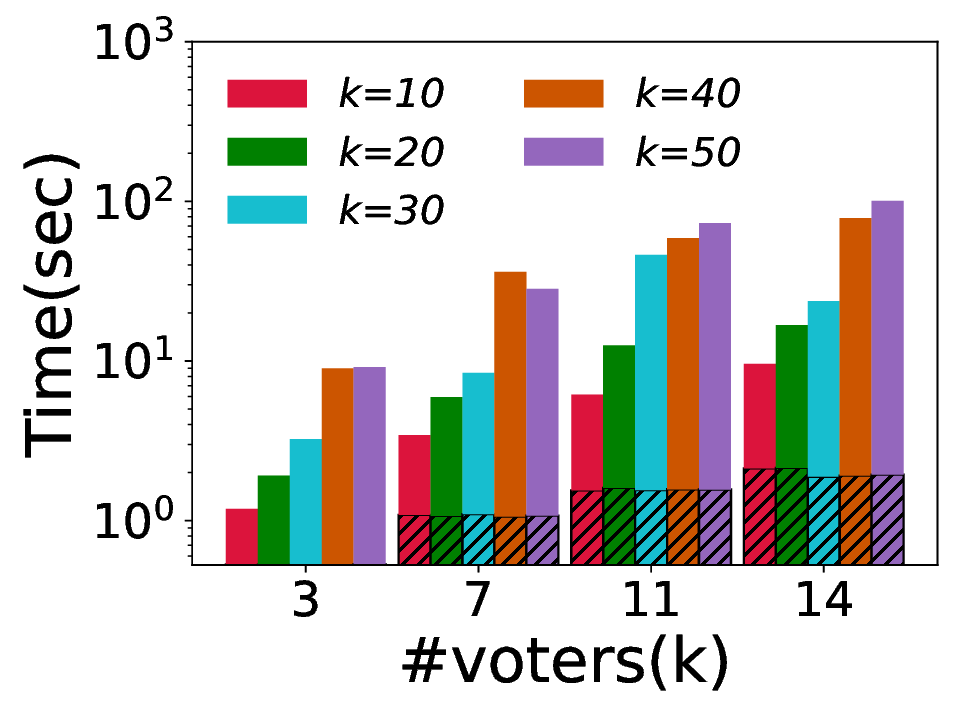}   
    \label{fig:commitee_size_trip_advisor_tgd}}
    \subfigure[Movies Dataset]{
    \centering
    \includegraphics[width=0.222\textwidth]{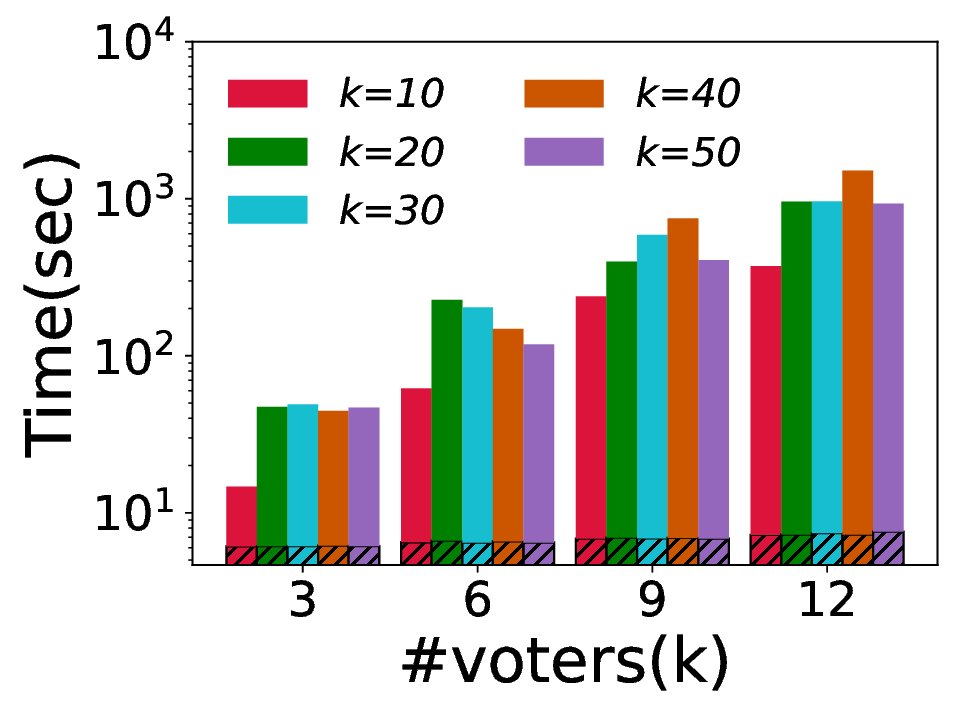}        
    \label{fig:commitee_size_movies_no_constraints}
    }
    \vskip-1em
    \caption{Computing time for varying committee size $k$. In the Trip Advisor dataset, the candidate group size is $|C|=1,845$, and we have one TGD, and for the Movies Dataset $|C|=100$ with no constraints.}
    \vskip-1em
\end{figure}

In \Cref{fig:commitee_size_trip_advisor_tgd} we see a direct continuation of the experiments in \Cref{fig:commitee_size}. In \Cref{fig:commitee_size_trip_advisor_tgd} we remove the DC from $\Gamma$, resulting in one TGD. This is to test whether the difference between \Cref{fig:trip_advisor_different_committee_size} and \Cref{fig:movies_different_committee_size} is due to the difference in $\Gamma$. In the first, there is a TGD and DC, and in the second there is only TGD (because adding the DC to the movies with a large committee size results in infeasible committees). After removing the DC we still get the same trend as in \Cref{fig:trip_advisor_different_committee_size}, therefore, we can conclude that the difference between \Cref{fig:trip_advisor_different_committee_size} and \Cref{fig:movies_different_committee_size} is not due to the difference in $\Gamma$.

In \Cref{fig:commitee_size_movies_no_constraints} we see a direct continuation of the experiments in \Cref{fig:commitee_size}. In \Cref{fig:commitee_size_movies_no_constraints} we remove the TGD from $\Gamma$, resulting in an empty set of constraints. This is to test whether the trend in \Cref{fig:movies_different_committee_size} is due to the TGD constraint. We get a slightly different result, but generally speaking, it still seems that the peek is not in $k=50$, but in $40$ or $20$.

In \Cref{fig:trip_advisor_different_committee_size}, there is a TGD and DC, and in \Cref{fig:movies_different_committee_size} there is only TGD (because adding the DC to the movies with a large committee size results in infeasible committees). After removing the DC we still get the same trend as in \Cref{fig:trip_advisor_different_committee_size}, therefore, we can conclude that the difference between \Cref{fig:trip_advisor_different_committee_size} and \Cref{fig:movies_different_committee_size} is not due to the difference in $\Gamma$.

In \Cref{fig:movies_different_constraints_multiple_constraints} we do similar experiments as \Cref{fig:dctgds:movies}, computing time for varying constraints number in the Movies Dataset. In \Cref{fig:movies_different_constraints_multiple_constraints} we add in addition to the default constraints of DC and TGD, two new cases, one 
DC with two TGDs, and one DC with three TGDs.

To do so, we first add to the database $D$ the relation $\rel{Duration}(c,t)$ that gives the duration range for each movie: \val{short}, \val{long} (derived based on median from the original dataset exact movie duration).
The two new TGD constraints are as follows.
\begin{itemize}
    \item There is at least one committee member for every movie duration category, \val{long} and \val{short}.
    \item There is at least one committee member for the genres of Comedy Drama and Action.
\end{itemize}

From this experiment, we learned that adding the third and fourth constraints had little to no impact on the running time compared to the case with two constraints.

\section{Case Study}\label{anacdotal_experiment}
In the experiment presented in \Cref{table:movies} we ran our MIP program in the Movies Dataset, where $k=5$, the voting rule is PAV and there are 3,000 voters. 
Here is a review of the results of this experiment, given different $\Gamma$ sets of constraints.

In the Movies Dataset, we define DC states that no two movies have the same genre (this DC is different from the default, where no three movies have the same genre because it is an experiment with a smaller committee than our default, $k=5$ instead of $k=10$). The TGD states that there is at least one movie in English, French, and Spanish (as the default).

When adding only a DC, we get a different committee. This is because the movies Judgment Night, The Dark, and Scarface all have the genre of a Thriller (a reminder that a movie can have multiple different genres). Therefore, the movies The Dark, and Scarface are changing. Furthermore, the movies 2001: A Space Odyssey and Three Colours: Red both have the Mystery genre. Therefore, 2001: A Space Odyssey is changed.

When adding only TGD, because we don't have representation for the Spanish language, we get that 2001: A Space Odyssey is replaced with Bad Education, a movie with Spanish as the original language.

Furthermore, when adding DC and TGD, we both need to solve the conflicts and have a representation of Spanish (without any conflicts). Therefore, the changes are different.

\section{Formal Phrasing of the Constraints}
As we explained in the Experimental Evaluation section, for each dataset we define two default constraints, a TGD and a DC. We also define two additional TGDs for the Movies Dataset for measuring the computation time for varying number of constraints. In addition, we have a simple DC for the Movies use case (since we define a smaller committee size) shown in \Cref{table:movies}. Here we present the formulation of all constraints in the formal framework (as opposed to the natural language in the body of the paper).

\def\colored#1{{\color{BlueViolet}#1}}

\subsection{Glasgow Dataset}
The TGD states that there is at least one committee member from each ward:
\colored{
\begin{align*}
\forall w\big[\rel{Wards}(w)\rightarrow\exists c 
    [\rel{Ward}(c,w)\land\comrel(c)]
\big]
\end{align*}
}

The DC states that no three committee members belong to the same party:
\colored{
\begin{align*}
\forall c_1,c_2,c_3,p 
\big[\neg\big(&\comrel(c_1)\land\comrel(c_2)\land\comrel(c_3)\land\\
     & \rel{Party}(c_1,p)\land\rel{Party}(c_2,p)\land\rel{Party}(c_3,p)\land\\
     &c_1\neq c_2 \land c_1\neq c_3 \land c_2\neq c_3\big)\big]
\end{align*}
}

\subsection{Trip Advisor Dataset}
For this TGD, we added a relation $\rel{Selected}(t,p)$ with six locations.
 
The TGD states that there is at least one low-price hotel in each selected location, for the selected city-country combinations:
\colored{
\begin{align*}
\forall l,p\big[\rel{Selected}(l,p)\rightarrow\exists c 
    [&\rel{Loacation}(c,l,p)\land\\&\rel{Price}(c,\val{low})\land\comrel(c)]
\big]
\end{align*}
}

The DC states that no two hotels have the same city, country, and price range:
 \colored{
 \begin{align*}
\forall c_1,c_2,l,p,r 
\big[\neg\big(&
\comrel(c_1)\land\comrel(c_2)\land c_1\neq c_2 \land
\\&
\rel{Location}(c_1,l,p)
\land\rel{Loacation}(c_2,l,p)\land
\\&
\rel{Price}(c_1,r)\land
\rel{Price}(c_2,r)\big)\big]
\end{align*}
}

\subsection{The Movies Dataset}

 The TGD states that there is at least one movie in English, French, and Spanish. For simplicity, we add the relation $\rel{SelectedLangs}$ that  contains the three languages. The TGD is then:

\colored{
\begin{align*}
\forall l\big[\rel{SelectedLangs}(l)\rightarrow\exists c 
    [\rel{Language}(c,l)\land\comrel(c)]
\big]
\end{align*}}
The DC states that no three movies have the same genre:
\colored{
 \begin{align*}
\forall c_1,c_2,c_3,g 
\big[\neg\big(&\comrel(c_1)\land\comrel(c_2)\land\comrel(c_3)\land\\
     & \rel{MovieGenre}(c_1,g)\land\rel{MovieGenre}(c_2,g)\land\\&\rel{MovieGenre}(c_3,g)\land c_1\neq c_2 \land c_1\neq c_3 \land \\&  c_2\neq c_3 \big)\big]
\end{align*}
}
\smallskip

The two additional TGDs are defined as follows. The relation $\rel{DurationCategory}$ contains the two length categories \val{long}, and \val{short} (based on the median).
The relation $\rel{SelectedGenres}$ contains the genres of Comedy, Drama, and Action.
The first TGD states that there is a movie from each duration category:
\colored{
\begin{align*}
\forall d\big[&\rel{DurationCategory}(d)\rightarrow\exists c [ \rel{Duration}(c,d) \land\comrel(c)
    ]
\big]
\end{align*}
}

\smallskip

The second TGD states that there is a movie from each selected genre:
\colored{
\begin{align*}
\forall g\big[\rel{SelectedGenres}(g)\rightarrow\exists c 
    [\rel{MovieGenre}(c,g) \land\comrel(c)]
\big]
\end{align*}
}

\smallskip

The additional DC, used in the case study, states that there no two movies in the same category:
\colored{
\begin{align*}
\forall c_1,c_2,g 
\big[\neg\big(&\comrel(c_1)\land\comrel(c_2)\land c_1\neq c_2 \land\\
     & \rel{MovieGenre}(c_1,g)\land\rel{MovieGenre}(c_2,g)  \big)\big]
\end{align*}
}



%% file: mcmf.pspdftex
\begin{picture}(0,0)%
\includegraphics{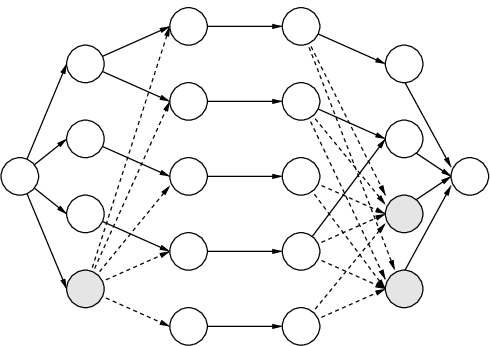}%
\end{picture}%
\setlength{\unitlength}{3947sp}%
\begin{picture}(3916,2755)(893,-2918)
\put(3301,-1636){\makebox(0,0)[b]{\smash{\fontsize{10}{12}\usefont{T1}{ptm}{m}{n}{\color[rgb]{0,0,0}$\cout 3$}%
}}}
\put(3301,-2836){\makebox(0,0)[b]{\smash{\fontsize{10}{12}\usefont{T1}{ptm}{m}{n}{\color[rgb]{0,0,0}$\cout 5$}%
}}}
\put(3301,-2236){\makebox(0,0)[b]{\smash{\fontsize{10}{12}\usefont{T1}{ptm}{m}{n}{\color[rgb]{0,0,0}$\cout 4$}%
}}}
\put(2851,-286){\makebox(0,0)[b]{\smash{\fontsize{10}{12}\usefont{T1}{ptm}{m}{n}{\color[rgb]{0,0,0}$\cst 1$}%
}}}
\put(2851,-1486){\makebox(0,0)[b]{\smash{\fontsize{10}{12}\usefont{T1}{ptm}{m}{n}{\color[rgb]{0,0,0}$\cst 3$}%
}}}
\put(2851,-886){\makebox(0,0)[b]{\smash{\fontsize{10}{12}\usefont{T1}{ptm}{m}{n}{\color[rgb]{0,0,0}$\cst 2$}%
}}}
\put(2851,-2086){\makebox(0,0)[b]{\smash{\fontsize{10}{12}\usefont{T1}{ptm}{m}{n}{\color[rgb]{0,0,0}$\cst 4$}%
}}}
\put(2851,-2686){\makebox(0,0)[b]{\smash{\fontsize{10}{12}\usefont{T1}{ptm}{m}{n}{\color[rgb]{0,0,0}$\cst 5$}%
}}}
\put(2401,-436){\makebox(0,0)[b]{\smash{\fontsize{10}{12}\usefont{T1}{ptm}{m}{n}{\color[rgb]{0,0,0}$\cin 1$}%
}}}
\put(2401,-1036){\makebox(0,0)[b]{\smash{\fontsize{10}{12}\usefont{T1}{ptm}{m}{n}{\color[rgb]{0,0,0}$\cin 2$}%
}}}
\put(2401,-1636){\makebox(0,0)[b]{\smash{\fontsize{10}{12}\usefont{T1}{ptm}{m}{n}{\color[rgb]{0,0,0}$\cin 3$}%
}}}
\put(2401,-2236){\makebox(0,0)[b]{\smash{\fontsize{10}{12}\usefont{T1}{ptm}{m}{n}{\color[rgb]{0,0,0}$\cin 4$}%
}}}
\put(2401,-2836){\makebox(0,0)[b]{\smash{\fontsize{10}{12}\usefont{T1}{ptm}{m}{n}{\color[rgb]{0,0,0}$\cin 5$}%
}}}
\put(4126,-1936){\makebox(0,0)[b]{\smash{\fontsize{10}{12}\usefont{T1}{ptm}{m}{n}{\color[rgb]{0,0,0}$w'_1$}%
}}}
\put(4126,-2536){\makebox(0,0)[b]{\smash{\fontsize{10}{12}\usefont{T1}{ptm}{m}{n}{\color[rgb]{0,0,0}$w'_2$}%
}}}
\put(4126,-736){\makebox(0,0)[b]{\smash{\fontsize{10}{12}\usefont{T1}{ptm}{m}{n}{\color[rgb]{0,0,0}$w_1$}%
}}}
\put(4126,-1336){\makebox(0,0)[b]{\smash{\fontsize{10}{12}\usefont{T1}{ptm}{m}{n}{\color[rgb]{0,0,0}$w_2$}%
}}}
\put(4651,-1636){\makebox(0,0)[b]{\smash{\fontsize{10}{12}\usefont{T1}{ptm}{m}{n}{\color[rgb]{0,0,0}$t'$}%
}}}
\put(1576,-736){\makebox(0,0)[b]{\smash{\fontsize{10}{12}\usefont{T1}{ptm}{m}{n}{\color[rgb]{0,0,0}$u_1$}%
}}}
\put(1576,-1336){\makebox(0,0)[b]{\smash{\fontsize{10}{12}\usefont{T1}{ptm}{m}{n}{\color[rgb]{0,0,0}$u_2$}%
}}}
\put(1576,-1936){\makebox(0,0)[b]{\smash{\fontsize{10}{12}\usefont{T1}{ptm}{m}{n}{\color[rgb]{0,0,0}$u_3$}%
}}}
\put(1576,-2536){\makebox(0,0)[b]{\smash{\fontsize{10}{12}\usefont{T1}{ptm}{m}{n}{\color[rgb]{0,0,0}$u'_1$}%
}}}
\put(1051,-1636){\makebox(0,0)[b]{\smash{\fontsize{10}{12}\usefont{T1}{ptm}{m}{n}{\color[rgb]{0,0,0}$s'$}%
}}}
\put(3301,-436){\makebox(0,0)[b]{\smash{\fontsize{10}{12}\usefont{T1}{ptm}{m}{n}{\color[rgb]{0,0,0}$\cout 1$}%
}}}
\put(3301,-1036){\makebox(0,0)[b]{\smash{\fontsize{10}{12}\usefont{T1}{ptm}{m}{n}{\color[rgb]{0,0,0}$\cout 2$}%
}}}
\end{picture}%